\documentclass[submission,copyright,creativecommons]{eptcs}
\providecommand{\event}{TLLA/Linearity 2018} 
\usepackage{underscore}           

\usepackage{marginnote}
\usepackage{microtype}
\usepackage{multicol}
\usepackage{mathtools,amsthm}
\usepackage{amsmath,amssymb}
\usepackage{stmaryrd}
\usepackage{xspace}
\usepackage{graphicx} 
\usepackage{cmll,ebproof}
\usepackage{ragged2e} 
\usepackage[all]{xy}

\hypersetup{hidelinks}
	
\newcommand{\Tot}{\mathsf{b}}

\newcommand{\ValScript}{\mathsf{v}}

\newcommand{\Weak}[1]{#1_\mathsf{w}}

\newcommand{\Betav}{\beta^v}

\newcommand{\Derel}{\mathsf{d}}

\newcommand{\Der}[1]{\mathrm{der}\,#1}
\newcommand{\Derp}[1]{\mathrm{der}(#1)}
\newcommand{\App}[2]{\langle #1 \rangle #2}
\newcommand{\La}[2]{\lambda #1 \, #2}
\newcommand{\Bang}[1]{#1^\oc}
\newcommand{\Bangp}[1]{(#1)^\oc}

\newcommand{\BangSet}{\oc\Lambda}
\newcommand{\BangSubset}[1]{\oc\Lambda_{#1}}

\newcommand{\ValueSet}{\oc\Lambda_\ValScript}
\newcommand{\VarSet}{\mathcal{V}\!\,ar}

\newcommand{\CtxSet}{\BangSubset{\Ctx}}
\newcommand{\WCtxSet}{\BangSubset{\WCtx}}
\newcommand{\LambdaVal}{\Lambda_v}
\newcommand{\LambdaCtx}{\Lambda_\LCtx}

\newcommand{\LL}{\mathsf{LL}}

\newcommand{\Fv}[1]{\mathsf{fv}(#1)}
\newcommand{\Sub}[2]{\{#1/#2\}}

\newcommand{\Rule}{\mathsf{r}}

\newcommand{\Root}[1]{\mapsto_{#1}}
\newcommand{\To}[1]{\to_{#1}}

\newcommand{\RevTo}[1]{{\,}_{#1}\!\!\leftarrow}

\newcommand{\MRevTo}[1]{{\,}_{#1}^*\!\!\leftarrow}
\newcommand{\WTo}[1]{\to_{\Weak{#1}}}

\newcommand{\ToBang}{\To{\oc}}
\newcommand{\ToVal}{\To{\ValScript}}

\newcommand{\WToBang}{\WTo{\oc}}

\newcommand{\WToVal}{\WTo{\ValScript}}

\newcommand{\ToTot}{\To{\Tot}}
\newcommand{\WToTot}{\WTo{\Tot}}

\newcommand{\Parallel}[1]{\Rightarrow_{#1}}
\newcommand{\ParallelVal}{\Parallel{\ValScript}}

\newcommand{\ToBeta}{\To{\beta}}
\newcommand{\ToBetav}{\To{\Betav}}
\newcommand{\WToBeta}{\WTo{\beta}}
\newcommand{\WToBetav}{\WTo{\Betav}}

\newcommand{\Var}{x}
\newcommand{\VarTwo}{y}
\newcommand{\VarThree}{z}
\newcommand{\var}{x}

\newcommand{\Tm}{T}
\newcommand{\TmTwo}{S}
\newcommand{\TmThree}{R}
\newcommand{\TmFour}{Q}
\newcommand{\TmFive}{P}

\newcommand{\Val}{V}
\newcommand{\ValTwo}{W}
\newcommand{\tm}{t}
\newcommand{\tmTwo}{s}
\newcommand{\tmThree}{r}
\newcommand{\tmFour}{q}
\newcommand{\tmFive}{p}

\newcommand{\ImTmOne}{M}
\newcommand{\ImTmTwo}{N}

\newcommand{\ImVal}{U}
\newcommand{\LVal}{V}
\newcommand{\LTm}{\Tm}
\newcommand{\LTmTwo}{\TmTwo}
\newcommand{\LTmThree}{\TmThree}
\newcommand{\LTmFour}{\TmFour}
\newcommand{\LTmFive}{\TmFive}

\newcommand{\Hole}[1]{\llparenthesis #1 \rrparenthesis}
\newcommand{\Ctx}{\mathtt{C}}

\newcommand{\WCtx}{\mathtt{W}}

\newcommand{\Ctxp}[1]{\Ctx\Hole{#1}}

\newcommand{\WCtxp}[1]{\WCtx\Hole{#1}}

\newcommand{\NCtx}{\mathtt{N}}
\newcommand{\VCtx}{\mathtt{V}}
\newcommand{\NCtxp}[1]{\NCtx\Hole{#1}}
\newcommand{\VCtxp}[1]{\VCtx\Hole{#1}}
\newcommand{\LCtx}{\Ctx}

\newcommand{\CBN}[1]{{#1^\mathsf{CBN}}}
\newcommand{\CBV}[1]{{#1^\mathsf{CBV}}}
\newcommand{\Cbn}[1]{{#1^\mathsf{cbn}}}
\newcommand{\Cbv}[1]{{#1^\mathsf{cbv}}}
\newcommand{\CbN}{CbN}
\newcommand{\CbV}{CbV}

\newcommand{\InvV}[1]{#1^\dagger}

\newcommand{\Defeq}{\coloneqq}
\newcommand{\Eqdef}{\eqqcolon}

\newcommand{\Refprop}[1]{Prop.\,\ref{prop:#1}}

\newcommand{\Refpropp}[2]{Prop.\,\ref{prop:#1}.\ref{prop:#1.#2}}

\newcommand{\Reflemma}[1]{Lemma~\ref{lemma:#1}}
\newcommand{\Reflemmap}[2]{Lemma~\ref{lemma:#1}.\ref{lemma:#1.#2}}

\newcommand{\Refthm}[1]{Thm.\,\ref{thm:#1}}
\newcommand{\Refthmp}[2]{Thm.\,\ref{thm:#1}.\ref{thm:#1.#2}}

\newcommand{\Refthmps}[3]{Thm.\,\ref{thm:#1}.\ref{thm:#1.#2}-\ref{thm:#1.#3}}
\newcommand{\Refcor}[1]{Cor.\,\ref{cor:#1}}
\newcommand{\Refcorp}[2]{Cor.\,\ref{cor:#1}.\ref{cor:#1.#2}}
\newcommand{\Refrmk}[1]{Rmk.\,\ref{rmk:#1}}

\newcommand{\Refex}[1]{Ex.\,\ref{ex:#1}}

\newcommand{\Reffig}[1]{Fig.\,\ref{fig-#1}}

\newcommand{\NoteProof}[1]{\marginpar{\RaggedRight \ \ \scriptsize{Proof p.\,{\pageref{#1}}}}}

\newcommand\ITens{\otimes}
\newcommand\Tens[2]{{#1}\ITens{#2}}
\newcommand\Tensp[2]{\left({#1}\ITens{#2}\right)}
\newcommand\IWith{\mathrel{\&}}
\newcommand\With[2]{{#1}\IWith{#2}}
\newcommand\IPlus{\oplus}

\newcommand\Orth[2][]{#2^{\Bot_{#1}}}
\newcommand\Orthp[2][]{(#2)^{\Bot_{#1}}}

\newcommand\cL{\mathcal L}

\newcommand\cU{\mathcal U}

\newcommand\One{1}
\newcommand\Leftu{\lambda}
\newcommand\Rightu{\rho}
\newcommand\Assoc{\alpha}
\newcommand\Sym{\sigma}

\newcommand\Limpl[2]{{#1}\multimap{#2}}

\newcommand\Evlin{\operatorname{\mathsf{ev}}}

\renewcommand{\Bot}{{\mathord{\perp}}}
\newcommand{\Top}{\top}
\newcommand{\Zero}{0}

\newcommand\Proj[1]{{\mathsf{pr}}_{#1}}
\newcommand\Inj[1]{{\mathsf{in}}_{#1}}

\newcommand\Excl[1]{\oc{#1}}

\newcommand\Dercat[1]{\operatorname{{\mathsf{der}}}_{#1}}

\newcommand\Diggcat[1]{\operatorname{{\mathsf{dig}}}_{#1}}

\newcommand\Seely{\mathsf{m}}
\newcommand\Seelyz{\Seely^0}
\newcommand\Seelyt{\Seely^2}

\newcommand\Curlin{\operatorname{\mathsf{cur}}}
\newcommand\Cur[1]{\Lambda(#1)}

\newcommand\Op[1]{{#1}^{\mathsf{op}}}
\newcommand\Kl[1]{{#1}_\oc}

\newcommand\Compl{\,}

\newcommand\Id{\operatorname{\mathsf{id}}}

\newcommand{\Ie}{\textit{i.e.}\xspace}
\newcommand{\Ih}{\textit{i.h.}\xspace}
\newcommand{\Eg}{\textit{e.g.}\xspace}
\newcommand{\Resp}{\textnormal{resp.}\xspace}
\newcommand{\Etc}{\textit{etc.}\xspace}

\newcommand\Isom\simeq
\newcommand\Vect[1]{\vec{#1}}
\newcommand\List[3]{#1_{#2},\dots,#1_{#3}}

\newcommand\Psem[1]{\llbracket{#1}\rrbracket}
\newcommand\Tpower[2]{{#1}^{\otimes{#2}}}

\newcommand\Contr[1]{\operatorname{\mathsf{c}}_{#1}}
\newcommand\Weakm[1]{\operatorname{\mathsf{w}}_{#1}}

\newcommand\Promp[1]{\left(#1\right)^!}
\newcommand\Zerom[2]{0_{#1,#2}}

\newcommand\Coalgm[1]{\mathsf h_{#1}}

\newcommand\Pair[2]{\langle{#1},{#2}\rangle}

\newcommand\Supp[1]{\mathsf{supp}(#1)}

\renewcommand{\NoteProof}[1]{
\marginnote{
\scriptsize{Proof p.\,{\pageref{#1}}}}}

\newcommand{\lto}{\multimap}
\renewcommand{\Val}{\Bang{\Tm}}
\renewcommand{\ValTwo}{\Bang{\TmTwo}}
\newcommand{\ValThree}{\Bang{\TmThree}}

\renewcommand{\LVal}{v}
\renewcommand{\LTm}{\tm}
\renewcommand{\LTmTwo}{\tmTwo}
\renewcommand{\LTmThree}{\tmThree}
\renewcommand{\LTmFour}{\tmFour}
\renewcommand{\LTmFive}{\tmFive}

\renewcommand{\ToBang}{\To{\Derel}}
\renewcommand{\WToBang}{\WTo{\Derel}}
\newcommand{\ToDer}{\ToBang}
\newcommand{\WToDer}{\WToBang}

\renewcommand{\ValScript}{\ell}

\renewcommand{\ValueSet}{\BangSet_\oc}
\renewcommand{\VarSet}{\mathcal{V}\!ar}

\newcommand{\valu}{box}

\newcommand{\Ground}{Ground}
\newcommand{\ground}{ground}

\renewcommand{\CBN}[1]{{#1^\mathsf{N}}}
\renewcommand{\CBV}[1]{{#1^\mathsf{V}}}
\renewcommand{\Cbn}[1]{{#1^\mathsf{n}}}
\renewcommand{\Cbv}[1]{{#1^\mathsf{v}}}

\renewcommand{\Weak}[1]{#1_\mathsf{g}}

\renewcommand{\WCtx}{\mathtt{G}}

\newcommand{\emptymset}{[\,]}

\newcommand{\seq}\vec
\newcommand\retract\vartriangleleft
\newcommand\lam{\ensuremath{\lambda}}
\newcommand\Impl[2]{{#1}\Rightarrow{#2}}
\newcommand\Klc{\circ}
\newcommand\Ev{\mathsf{Ev}}
\newcommand{\Vdashv}{\vdash_\textup{v}}
\newcommand{\Intv}[1]{\vert#1\vert^\textup{v}}
\newcommand{\Intn}[1]{\vert#1\vert^\textup{n}}

\newcommand{\app}[1]{\mathsf{app}_{\sf #1}}
\newcommand{\Lam}[1]{\mathsf{lam}_{\sf #1}}
\newcommand{\Emptymset}{[\,]}

\newcommand{\lincalc}{$\lambda_\mathsf{lin}$}

\let\Gamma\varGamma
\let\Delta\varDelta
\let\Theta\varTheta
\let\Lambda\varLambda
\let\Xi\varXi
\let\Pi\varPi
\let\Sigma\varSigma
\let\Upsilon\varUpsilon
\let\Phi\varPhi
\let\Psi\varPsi
\let\Omega\varOmega
\renewcommand{\NoteProof}[1]{}

\newtheorem{definition}{Definition} 
\newtheorem{lemma}[definition]{Lemma} 
\newtheorem{proposition}[definition]{Proposition} 
\newtheorem{theorem}[definition]{Theorem} 
\newtheorem{corollary}[definition]{Corollary} 

\theoremstyle{remark}
\newtheorem{remark}[definition]{Remark} 
\newtheorem{example}[definition]{Example} 

\title{The Bang Calculus and the Two Girard's Translations}
\author{Giulio Guerrieri
\institute{Dipartimento di Informatica -- Scienza e Ingegneria (DISI), Universit\`a di Bologna, Bologna, Italy}
\email{\href{mailto:giulio.guerrieri@unibo.it}{giulio.guerrieri@unibo.it}}
\and
Giulio Manzonetto
\institute{LIPN, UMR 7030, Universit\'e Paris 13, Sorbonne Paris Cit\'e, F-93430, Villetaneuse, France}
\email{\href{mailto:giulio.manzonetto@lipn.univ-paris13.fr}{giulio.manzonetto@lipn.univ-paris13.fr}}
}
\def\titlerunning{The Bang Calculus and the Two Girard's Translations}
\def\authorrunning{G. Guerrieri \& G. Manzonetto}

\begin{document}
\maketitle

\begin{abstract}
We study the two Girard's translations of intuitionistic implication into linear logic by exploiting the bang calculus, a paradigmatic functional language with an explicit box-operator that 
allows both call-by-name and call-by-value $\lam$-calculi to be encoded in.
We investigate how the bang calculus subsumes both call-by-name and call-by-value \lam-calculi from a syntactic and a semantic viewpoint.
\end{abstract}

\section{Introduction}
\label{sect:intro}

The $\lambda$-calculus is a simple 
framework formalizing many features of functional programming languages. 
For instance, 
the $\lambda$-calculus can be endowed with two distinct 
evaluation mechanisms (among others), \emph{call-by-name} (\CbN) and \emph{call-by-value} (\CbV), 
having quite 
different properties. 
A \CbN\ discipline re-evaluates an argument each time it is used. 
By contrast, a \CbV\ discipline first evaluates an argument 
once and for all, 
then recalls its value whenever required.
\CbN\ and \CbV\ $\lambda$-calculi are usually defined by means of 
operational rules giving rise to two different rewriting systems on the same set of $\lambda$-terms: in \CbN\ there is no restriction on firing a $\beta$-redex, whereas in \CbV\ a $\beta$-redex can be fired only when the argument is a value, 
\Ie a variable or an abstraction.
The standard categorical setting for describing denotational models 
of the $\lambda$-calculus, cartesian closed categories, provides models which 
are adequate for \CbN, but typically not for \CbV. 
For \CbV, the introduction of an additional computational monad (in the sense of Moggi \cite{Moggi89,Moggi91}) is necessary. 
While \CbN\ $\lambda$-calculus \cite{Barendregt84} has a rich and 
refined semantic and syntactic theory featuring advanced concepts such as 
separability, solvability, B\"ohm trees, classification of $\lambda$-theories, 
full-abstraction,  \textit{etc.}, this is not the case for \CbV\ 
$\lambda$-calculus \cite{Plotkin75}, in the sense that concerning the \CbV\ 
counterpart of these theoretical notions 
there are only \mbox{partial and not satisfactory results (or they do not exist at 
all!).} 

Quoting from \cite{Levy99}, ``the existence of two separate paradigms is troubling'' for at least two reasons:
\begin{itemize}
  \item it makes each language appear arbitrary (whereas a unified language might be more canonical);
  \item 
  each time we create a new style of semantics, \Eg Scott semantics, operational semantics, game semantics, continuation semantics, \Etc, we always need to do it twice\,---\,once for each paradigm.
\end{itemize}

Girard's Linear Logic ($\LL$, \cite{Girard87}) 
provides a unifying setting where this discrepancy could be solved since both \CbN\ and \CbV\ $\lambda$-calculi can be faithfully translated, via two different translations, into $\LL$ proof-nets.
Following \cite{Levy99}, we can claim that, via these translations, $\LL$ proof-nets ``subsume'' the \CbN\ and \CbV\ paradigms, in the sense that both operational and denotational semantics for those paradigms can be seen as arising, via these translations, from similar semantics for $\LL$.

Indeed, $\LL$ can be understood as a refinement of intuitionistic logic (and hence
$\lambda$-calculus) in which resource management is made explicit thanks to the 
introduction of a new pair of dual connectives: the exponentials ``$\oc$'' and ``$\wn$''. 
In proof-nets, the standard syntax for $\LL$ proofs, 
\emph{boxes} (introducing the modality ``$\oc$'') mark 
the sub-proofs available at will: during cut-elimination, such boxes can be 
erased (by weakening rules), can be duplicated (by contraction rules), can be opened (by dereliction rules) or can enter other 
boxes.
The categorical counterpart of this refinement is well known: it is the notion
of a cartesian $*$-autonomous\footnote{Actually the full symmetry of such a
category is not really essential as far as the $\lambda$-calculus is concerned,
it is however quite natural from the $\LL$ viewpoint: $\LL$ restores the
classical involutivity of negation in a constructive setting.} category,
equipped with a comonad endowed with a strong monoidal structure. 
Every instance of such a kind of structure yields a denotational model of $\LL$.

In his seminal article~\cite[p.~78]{Girard87}, Girard proposes a standard translation
of intuitionistic logic (and hence simply typed $\lambda$-calculus) in 
multiplicative-exponential $\LL$ proof-nets whose semantic counterpart is well
known: the Kleisli category of the exponential comonad ``$\oc$'' is cartesian 
closed thanks to the strong monoidal structure of ``$\oc$''. 
This translation $\CBN{(\cdot)}$ maps the intuitionistic implication 
$A\Rightarrow B$ to the $\LL$ formula $\Limpl{\Excl{\CBN{A}}}{\CBN{B}}$.
In \cite[p.~81]{Girard87} Girard proposes also another translation $\CBV{(\cdot)}$ that he calls ``boring'': 
it maps the the intuitionistic implication $A \Rightarrow B$ to the $\LL$ formula $\Excl{(\Limpl{\CBV A}{\CBV B})}$ (or equivalently $\Limpl{\Excl{\CBV A}}{\Excl{\CBV B}}$).
Since the untyped $\lambda$-calculus can be seen as simply typed with only one ground type $o$ satisfying the recursive identity $o = o \Rightarrow o$, the two Girard's translations $\CBN{(\cdot)}$ and $\CBV{(\cdot)}$ decompose this identity into $o = \oc o \multimap o$ and $o = \oc(o \multimap o)$ (or equivalently, $o = \oc o \multimap \oc o$), respectively.
At the $\lambda$-term level, these two translations differ only by the way they use logical 
exponential rules (\Ie box and dereliction), whereas 
they use multiplicative and structural (\Ie contraction and weakening) ingredients in the same way.
Because of this difference, the translation $\CBN{(\cdot)}$ encodes the \CbN\ $\lam$-calculus into $\LL$ proof-nets (in the sense that \CbN\ evaluation $\To{\beta}$ is simulated by cut-elimination via $\CBN{(\cdot)}$), while $\CBV{(\cdot)}$ encodes the \CbV\ $\lambda$-calculus into $\LL$ proof-nets (\CbV\ evaluation $\To{\Betav}$ is simulated by cut-elimination via $\CBV{(\cdot)}$).
Indeed, since in \CbN\ $\lambda$-calculus there is
no restriction on firing a $\beta$-redex 
(its argument can be freely copied or erased), the translation $\CBN{(\cdot)}$ 
puts the argument of every application into a box (see \cite{Danos90,Regnier92,Laurent03}); on the other hand, the translation 
$\CBV{(\cdot)}$ puts only values into boxes (see \cite{Accattoli15}) since in \CbV\ $\lambda$-calculus values are the only duplicable and discardable $\lambda$-terms.
Thus, as deeply studied in \cite{MaraistOderskyTurnerWadler99}, the two Girard's
logical translations explain the two different evaluation mechanisms, bringing them into the \mbox{scope of the Curry-Howard isomorphism.}

The syntax of multiplicative-exponential $\LL$ proof-nets is extremely expressive and powerful, but it is too general and sophisticated for the computational purpose of representing purely functional programs. 
For instance, simulation of $\beta$-reduction on $\LL$ proof-nets passes through intermediate states/proof-nets that cannot be expressed as $\lambda$-terms, since $\LL$ proof-nets 
have many spurious cuts with axioms that have no counterpart on $\lam$-terms.
More generally, $\LL$ proof-nets are manipulated in their graphical form, and while this is a handy formalism for intuitions, it is far from practical for formal reasoning.

From the analysis of Girard's translations it seems worthwhile to extend the syntax of the $\lambda$-calculus 
to internalize the insights coming from $\LL$ \emph{in a $\lambda$-like syntax}.
The idea is to enrich the $\lambda$-calculus 
with explicit \emph{boxes} 
marking the ``values'' of the calculus, \Ie the terms that 
can be freely duplicated and discarded: such a \emph{linear} $\lambda$-calculus 
subsumes both \CbN\ and \CbV\ $\lambda$-calculi, via suitable translations. 
This, of course, has been done quite early in the history of $\LL$ by defining various linear $\lambda$-calculi, 
such as \cite{LincolnMitchell92,Abramsky93,BentonBiermanDePaivaHyland93,BentonWadler96,RonchiRoversi97,MaraistOderskyTurnerWadler99, DBLP:conf/rta/Simpson05}.
All these calculi require a clear distinction between linear and non-linear variables, 
structural rules being freely (and implicitly) available for the latter and forbidden for the former. 
This distinction complicates the 
formalism and is actually useless as far as we \mbox{are interested in subsuming $\lambda$-calculi.}

Inspired by Ehrhard \cite{Ehrhard16}, in \cite{EhrhardG16} it has been introduced an intermediate formalism enjoying at the
same time the conceptual simplicity of $\lambda$-calculus (without any distinction between linear and non-linear variables) and the operational
expressiveness of $\LL$ proof-nets: the \emph{bang calculus}. 
It is a variant of the $\lambda$-calculus which is ``linear'' in the sense that the
exponential rules of $\LL$ (box and dereliction) are part of the syntax, so as to subsume \CbN\ and \CbV\ $\lambda$-calculi via two translations $\Cbn{(\cdot)}$ and $\Cbv{(\cdot)}$, respectively, from the set $\Lambda$ of $\lambda$-terms to the set $\BangSet$ of terms of the bang calculus (see \S\ref{sect:embedding}).
These two translations are deeply related to Girard's encodings $\CBN{(\cdot)}$ and $\CBV{(\cdot)}$ of \CbN\ and \CbV\ $\lambda$-calculi into $\LL$ proof-nets. 
Indeed, Girard's translations $\CBN{(\cdot)}$ and $\CBV{(\cdot)}$ decompose in such a way 
that the following diagrams commute: 

{\small
\begin{equation*}
  \xymatrix@=10pt{
    \Lambda \ar[rrrrrr]|-{\ \CBN{(\cdot)}\ } \ar[drrr]|-{\ \Cbn{(\cdot)}\ } & & & & & & \LL \\
    & & & \BangSet \ar[urrr]|-{\ (\cdot)^\circ\ }
  }
  \qquad\qquad
  \xymatrix@=10pt{
    \Lambda \ar[rrrrrr]|-{\ \CBV{(\cdot)}\ } \ar[drrr]|-{\ \Cbv{(\cdot)}\ } & & & & & & \LL \\
    & & & \BangSet \ar[urrr]|-{\ (\cdot)^\circ\ }
  }
\end{equation*}
}%

\noindent where $(\cdot)^\circ$ is a natural translation of the bang calculus into multiplicative-exponential $\LL$ proof-nets.
Thus, the bang calculus \emph{internalizes} the two Girard's translations in a $\lambda$-like calculus instead of $\LL$ proof-nets.
It subsumes both \CbN\ and \CbV\ $\lambda$-calculi in the \emph{same} rewriting system and denotational model, so that it may be a general setting to compare \CbN\ and \CbV.
The bang calculus can be seen as a metalanguage where the choice of \CbN\ or \CbV\ evaluation depends on the way the term is built up.
If we consider the syntax of the $\lambda$-calculus as a programming language, issues like \CbN\ versus \CbV\ evaluations affect the way the $\lambda$-calculus is translated in this metalanguage, but does not affect the metalanguage itself.


It turns out that this bang calculus was already known in the literature: it is 
an untyped version of the implicative fragment of Paul Levy's Call-By-Push-Value calculus \cite{Levy99,Levy06}.
Interestingly, his work was not motivated by an investigation of the two Girard's translations. 
This link is not casual, since it holds even when the bang calculus is extended to a PCF-like system, as shown by Ehrhard \cite{Ehrhard16}.

The aim of our paper is to further investigate the way the bang calculus subsumes \CbN\ and \CbV\ $\lambda$-calculi, refining and extending some results already obtained in \cite{EhrhardG16}.

\begin{enumerate}
  \item From a syntactic viewpoint, we show in \S\ref{sect:embedding} that the bang calculus subsumes in the same rewriting system both \CbN\ and \CbV\ $\lambda$-calculi, in the sense that the translations $\Cbn{(\cdot)}$ and $\Cbv{(\cdot)}$ from the $\lambda$-calculus to the bang calculus are sound and complete with respect to $\beta$-reduction and $\beta_v$-reduction, respectively (in \cite{EhrhardG16} only soundness was proven, and in a less elegant way). 
  \mbox{In other words, the diagrams}
  
  \vspace{-\baselineskip}
  {\small
\begin{equation*}
  \xymatrix@=10pt{
    \Lambda \ni \LTm \ar[rrrr]|-{\ \beta \ } \ar[dd]^{\ \Cbn{(\cdot)}\ } & & & & \LTmTwo \in \Lambda \ar[dd]^{\ \Cbn{(\cdot)}\ } \\\\
    \BangSet \ni \Cbn{\LTm} \ar[rrrr]|-{\ \Tot \ }  & & & & \Cbn{\LTmTwo} \in \BangSet
  }
  \qquad\qquad
  \xymatrix@=10pt{
    \Lambda \ni \LTm \ar[rrrr]|-{\ \beta_v \ } \ar[dd]^{\ \Cbv{(\cdot)}\ } & & & & \LTmTwo \in \Lambda \ar[dd]^{\ \Cbv{(\cdot)}\ } \\\\
    \BangSet \ni \Cbv{\LTm} \ar[rrrr]|-{\ \Tot \ }  & & & & \Cbv{\LTmTwo} \in \BangSet
  }
\end{equation*}
}%

  \noindent commute 
  in the two ways: starting from the $\beta$-reduction step $\To{\beta}$ for the \CbN\ $\lambda$-calculus (on the left) or the $\beta_v$-reduction step $\To{\Betav}$ for the \CbV\ $\lambda$-calculus (on the right), and starting from the $\Tot$-reduction step $\ToTot$ of the bang calculus.%
  \footnote{Actually, for the \CbV\ $\lambda$-calculus the diagram is slightly more complex, as we will see in \S\ref{sect:embedding}, but the essence does not change.}
  
  \item From a semantic viewpoint, we show in \S\ref{sect:semantics} that  \emph{every} $\LL$-based model $\cU$ of the bang calculus (as categorically defined in \cite{EhrhardG16}) provides a model for both \CbN\ and \CbV\ $\lambda$-calculi (in \cite{EhrhardG16} this was done only for the special case of relational semantics). 
  Moreover, given a $\lambda$-term $\LTm$, we investigate the relation between its interpretations $\Intn{\LTm}$ in \CbN\ (resp.~$\Intv{\LTm}$ in \CbV) and the interpretation $\Psem{\cdot}$ of its translation $\Cbn{\LTm}$ (resp.~$\Cbv{\LTm}$) into the bang calculus. 
  We prove that the diagram below on the left (for \CbN) commutes, whereas we give a counterexample (in the relational semantics) to the commutation of the diagram below on the right (for \CbV).
  We conjecture that there still exists a relationship in \CbV\ between $\Intv{\LTm}$ and $\Psem{\Cbv{\LTm}}$, but it should be more sophisticated than in \CbN.
  
  \vspace{-\baselineskip}
  {\small
  \begin{equation*}
    \xymatrix@=10pt{
      \Lambda \ni \LTm \ar[rrrr]|-{\ \Intn{\cdot}\ } \ar[drr]|-{\,\Cbn{(\cdot)}} & & & & \Intn{\LTm} = \Psem{\Cbn{\LTm}} \in \cU \\
      & &  \Cbn{\LTm} \in \BangSet \ar[urr]|-{\ \Psem{\cdot}\ }
    }
    \qquad\quad
    \xymatrix@=10pt{
      \Lambda \ni \LTm \ar[rrrr]|-{\ \Intv{\cdot}\ } \ar[drr]|-{\,\Cbv{(\cdot)}} & & & & \Intv{\LTm} = \Psem{\Cbv{\LTm}} \in \cU \\
      & & \Cbv{\LTm} \in \BangSet \ar[urr]|-{\ \Psem{\cdot}\ }
    }
  \end{equation*}
  }%

\end{enumerate}

In order to achieve these results in a clearer and simpler way, we have slightly modified (see \S\ref{sect:calculus}) the syntax and operational semantics of the bang calculus with respect to its original formulation in \cite{EhrhardG16}. 

{\small
\paragraph{Preliminaries and notations.}
  Let $\To{\Rule}$ and $\To{\Rule'}$ be binary relations on a set $X$.
		The composition of $\To{\Rule}$ and $\To{\Rule'}$ is denoted by $\To{\Rule}\To{\Rule'}$ or $\To{\Rule} \!\cdot\! \To{\Rule'}$.
    The transpose of $\To{\Rule}$ is denoted by $\RevTo{\Rule}$.
    The reflexive-transitive (resp.~reflexive) closure of $\To{\Rule}$ is denoted by $\To{\Rule}^*$ (resp.~$\To{\Rule}^=$).
    The \emph{$\Rule$-equivalence} $\simeq_\Rule$ is the 
    reflexive-transitive and symmetric closure of $\To{\Rule}$.
		Let $\LTm \in X$: $\LTm$ is \emph{$\Rule$-normal} if there is no $\LTmTwo \in X$ such that $\LTm \to_\Rule \LTmTwo$; \, $\LTm$ is \emph{$\Rule$-normalizable} if there is a $\Rule$-normal $\LTmTwo \in X$ such that $\LTm \to_\Rule^* \LTmTwo$, and we then say that $\LTmTwo$ is a \emph{$\Rule$-normal form of $\LTm$}. 
    
    The relation $\To{\Rule}$ is \emph{confluent} if $\MRevTo{\Rule} \cdot\! \To{\Rule}^* \ \subseteq \ \To{\Rule}^* \!\cdot \MRevTo{\Rule}$; it is \emph{quasi-strongly confluent} if $\RevTo{\Rule} \cdot\! \To{\Rule} \ \subseteq \ (\To{\Rule} \!\cdot \RevTo{\Rule}) \,\cup =$.
    From confluence it follows that: $\LTm \simeq_{\Rule} \LTmTwo$ iff $\LTm \To{\Rule}^* \LTmThree \,\,{}_\Rule^*\!\!\!\leftarrow \LTmTwo$ for some $\LTmThree \in X$; 
    and every $\Rule$-normalizable $\LTm \in X$ has a \emph{unique} $\Rule$-normal form.
		Clearly, quasi-strong confluence implies confluence.
}

\section{Syntax and reduction rules of the bang calculus}
\label{sect:calculus}

The syntax and operational semantics of the \emph{bang calculus} are defined in \Reffig{bangcalculus}.

\begin{figure}[!t]
  \centering
  \scalebox{0.9}{\parbox{1.05\linewidth}{
  \begin{align*}
    \text{\emph{Terms}:}		&& \Tm, \TmTwo, \TmThree &\Coloneqq \, \Var  \,\mid\, \La{\Var}\Tm \,\mid\, \App{\Tm}{\TmTwo} \,\mid\, \Der{\Tm}	\,\mid \, \Val			&&(\textup{set: } \BangSet)  \\
    \text{\emph{Contexts}:}	&& \Ctx			&\Coloneqq \, \Hole{\cdot} \mid \La{\Var}{\Ctx} \mid \App{\Ctx}{\Tm} \mid \App{\Tm}{\Ctx} \mid \Der{\Ctx} \mid \Bang{\Ctx}	&&(\textup{set: }\CtxSet) \\
    \text{\emph{\Ground\ contexts}:}&&\WCtx  		&\Coloneqq \, \Hole{\cdot} \mid \La{\Var}{\WCtx} \mid \App{\WCtx}{\Tm} \mid \App{\Tm}{\WCtx}  \mid \Der{\WCtx} 			&&(\textup{set: }\WCtxSet)	     \\[-2\baselineskip]
  \end{align*}
  \begin{align*}
    \text{\emph{Root-steps}:}	&& \App{\La{\Var}{\Tm}}{\ValTwo} &\Root{\ValScript} \Tm\Sub{\TmTwo}{\Var}  \qquad \Derp{\Bang{\Tm}} \Root{\Derel} \Tm \qquad \Root{\Tot} \, \Defeq \ \Root{\ValScript} \cup \Root{\Derel} \\[.2\baselineskip]
    \Rule\text{-\emph{reduction}:} &&   \Tm \To{\Rule}  \TmTwo \ &\Leftrightarrow \ \exists \, \Ctx \in \CtxSet, \, \exists\, \Tm'\!, \TmTwo' \!\in \BangSet : \Tm = \Ctxp{\Tm'}, \, \TmTwo = \Ctxp{\TmTwo'}, \, \Tm' \!\Root{\Rule} \TmTwo' \\
    \Weak{\Rule}\text{-\emph{reduction}:} &&   \Tm \WTo{\Rule}  \TmTwo \ &\Leftrightarrow \ \exists \, \WCtx \in \WCtxSet, \,\exists\, \Tm'\!, \TmTwo' \!\in \BangSet : \Tm = \WCtxp{\Tm'}, \, \TmTwo = \WCtxp{\TmTwo'}, \, \Tm' \!\Root{\Rule} \TmTwo'
    \\[-1.5\baselineskip]
  \end{align*}
  }}
  \caption{The bang calculus: its syntax and its reduction rules, where $\Rule \in \{\ValScript, \Derel, \Tot\}$.}
  \label{fig-bangcalculus}
\end{figure}

Terms are built up from a countably infinite set $\VarSet$ of \emph{variables} (denoted by $\Var, \VarTwo, \VarThree, \dots$). 
Terms of the form $\Bang{\Tm}$ (resp.~$\lambda\Var \, {\Tm}$; $\App{\Tm}{\TmTwo}$; $\Der{\Tm}$) are called 
or \emph{boxes} (resp.~\emph{abstractions}; \emph{(linear) applications}; \emph{derelictions}). 
The set of boxes is denoted by $\ValueSet
$.
The set of free variables of a term $\Tm$, denoted by $\Fv{\Tm}$, is defined as expected, $\lambda$ being the only binding construct. 
All terms are considered up to $\alpha$-conversion. 
Given $\Tm, \TmTwo \in \BangSet$ and a variable $\Var$, $\Tm \Sub{\TmTwo}{\Var}$ denotes the term obtained by the \emph{capture-avoiding substitution} of $\TmTwo$ (and not $\ValTwo$) for each free occurrence of $\Var$ in $\Tm$: so, $\Bang{\Tm} \Sub{\TmTwo}{\Var} = \Bang{(\Tm \Sub{\TmTwo}{\Var})} \in \ValueSet$.

\emph{Contexts} $\Ctx$ and \emph{\ground\ contexts} $\WCtx$ (both with exactly one hole $\Hole{\cdot}$) are defined in \Reffig{bangcalculus}.
All ground contexts are contexts but the converse fails: $\Bang{\Hole{\cdot}}$ is a non-\ground\ context. 
We write $\Ctxp{\Tm}$ for the term obtained by the capture-allowing substitution of the term $\Tm$ for the hole $\Hole{\cdot}$ in the context $\Ctx$.

Reductions in the bang calculus are defined in \Reffig{bangcalculus} as follows: 
given a \emph{root-step} rule $\Root{\Rule} \, \subseteq \BangSet \times \BangSet$, 
we define the \emph{$\Rule$-reduction} $\To{\Rule}$ (resp.~\emph{$\Weak{\Rule}$-reduction} or \emph{\ground\ $\Rule$-reduction} $\WTo{\Rule}$) as the closure of $\Root{\Rule}$ under contexts (resp.~\ground\ contexts).
Note that $\WTo{\Rule} \, \subsetneq \, \To{\Rule}$ as $\WCtxSet \subsetneq \CtxSet$: 
the only difference between $\To{\Rule}$ and $\WTo{\Rule}$ is that the latter does not reduce under $\oc$ (but both reduce under $\lambda$).
The root-steps used in the bang calculus are $\Root{\ValScript}$ and $\Root{\Derel}$ and $\Root{\Tot} \, \Defeq \ \Root{\ValScript} \!\cup \Root{\Derel}$.
From the definitions in \Reffig{bangcalculus} it follows that $\ToTot \, = \, \ToVal \cup \ToBang$ and $\WToTot \, = \, \WToVal\! \cup \WToBang$.
In $\LL$ proof-nets, $\Tot$-reduction and $\Weak{\Tot}$-reduction correspond to cut-elimination and cut-elimination outside boxes, respectively.

Intuitively, the basic idea behind the root-steps $\Root{\ValScript}$ and $\Root{\Derel}$ is that the \valu-construct $\oc$ marks the only 
terms that can be erased and duplicated.
When the argument of a construct $\Der{}$ is a box $\Val$, the root-step $\Root{\Derel}$ opens the box, \Ie accesses its content $\Tm$, 
destroying its status of availability at will (but 
$\Tm$, in turn, might be a box).
The root-step $\Root{\ValScript}$ says that a $\beta$-like redex $\App{\La{\Var}{\Tm}}{\TmTwo}$ can be fired only when its argument is a \valu, \Ie$\TmTwo = \Bang{\TmThree}$: 
if it is so, the content $\TmThree$ of the box $\TmTwo$ replaces any free occurrence of $\Var$ in $\Tm$.\footnotemark
\footnotetext{In \cite{EhrhardG16}, the definition of 
	$\Root{\ValScript}$ is slightly different from \Reffig{bangcalculus}: 
	$\App{\La{\Var}{\Tm}}{V} \Root{\ValScript} \Tm\Sub{V\!}{\Var}$ where $V$ is a variable or a box.
	Logically, this means that a variable of the bang calculus corresponds in $\LL$ proof-nets to an exponential axiom 
	in \cite{EhrhardG16}, and to a derelicted axiom here.
	The two definitions $\Root{\ValScript}$ of are expressively equivalent (they can be simulated each other), but the one adopted here allows for more elegant embeddings of \CbN\ and \CbV\ $\lambda$-calculi into the bang calculus (\emph{cf.}~\Refthm{embedding} below with Prop.~2 in \cite{EhrhardG16}).}

\begin{example}
\label{ex:delta}
  Let $\Delta \Defeq \La{\Var}{\App{\Var}{\Bang{\Var}}}$ and $\Delta\!' \Defeq \La{\Var}\App{\Derp{\Bang{\Var}}}{\Bang{\Var}}$. 
  Then, $\Delta\!' \WToBang \Delta$ and $\App{\Delta}{\Bang{\Delta\!}} \WToVal \!\App{\Delta}{\Bang{\Delta\!}} \WToVal\!\dots$ and $\App{\Derp{\Bang{\Delta\!'}}}{\Bang{\Delta\!'}} \!\WToBang\! \App{\Delta\!'}{\Bang{\Delta\!'}} \!\WToVal\! \App{\Derp{\Bang{\Delta\!'}}}{\Bang{\Delta\!'}} \WToBang\!\dots$ 
  Note that $\Bang{(\App{\Delta}{\Bang{\Delta\!}})}$ is $\Weak{\Tot}$-normal but not $\Tot$-normalizable.
\end{example}


%


The \emph{bang} (resp.~\emph{ground bang}) \emph{calculus} is the set $\BangSet$ endowed with the reduction $\ToTot$ (resp.~$\WToTot$).

\paragraph{\texorpdfstring{Quasi-strong confluence of $\Weak{\Tot}$-reduction and confluence of $\Tot$-reduction.}{Quasi-strong confluence of bw-reduction and confluence of b-reduction}}

To prove the confluence of $\ToTot$ (\Refpropp{confluence}{ToTot}), first we show that $\ToVal$ is confluent (\Reflemmap{basic-ToBang-ToVal}{ToVal-confluent}). 
The latter is proved by a standard adaptation of Tait--Martin-L\"of technique\,---\,as improved by Takahashi \cite{Takahashi95}\,---\,based on parallel reduction. 
For this purpose, we introduce \emph{parallel $\ValScript$-reduction}, denoted by $\ParallelVal$, a binary relation on $\BangSet$ defined by the rules in Fig.~\ref{fig:parallel}.
Intuitively, $\ParallelVal$ reduces simultaneously a number of $\ValScript$-redexes existing in a term.
It is immediate to check that $\ParallelVal$ is reflexive and $\ToVal \, \subseteq \, \ParallelVal \, \subseteq \, \ToVal^*$, hence $\ParallelVal^* \, = \, \ToVal^*$\,.

For any term $\Tm$, we denote by $\Tm^*$ the term obtained by reducing \emph{all} $\ValScript$-redexes in $\Tm$ simultaneously.
Formally, $\Tm^*$ is defined by induction on $\Tm \in \BangSet$ as follows:
\begin{equation*}
\begin{gathered}
  \Var^* \Defeq \Var \qquad\quad (\La{\Var}{\Tm})^* \Defeq \La{\Var}{\Tm^*} \qquad\quad (\Bang{\Tm})^* \Defeq \Bang{(\Tm^*)} \qquad\quad (\Der{\Tm})^* \Defeq \Derp{\Tm^*}	\\
  (\App{\Tm}{\TmTwo})^* \Defeq \App{\Tm^*}{\TmTwo^*} \textup{ \ if } \Tm \neq \La{\Var}{\TmThree} \textup{ or } \TmTwo \notin \ValueSet \qquad\qquad\qquad\qquad
   (\App{\La{\Var}{\Tm}}{\ValTwo})^* \Defeq \Tm^*\Sub{{\TmTwo}^*}{\Var}\,.
\end{gathered}
\end{equation*}

\newcounter{lemma:development}
\addtocounter{lemma:development}{\value{definition}}
\begin{lemma}[Development]
  \label{lemma:development}
  Let 
  \NoteProof{lemmaAppendix:development}
  $\Tm, \TmTwo \in \BangSet$.
  If $\Tm \ParallelVal \TmTwo$, then $\TmTwo \ParallelVal \Tm^*$.
\end{lemma}

\begin{figure}[!bt]
  \centering
  \scalebox{0.85}{
    \begin{prooftree}[label separation=0.3em]
      \infer0[\footnotesize$\mathsf{var}$]{\Var \ParallelVal \Var}
    \end{prooftree}
    \quad
    \begin{prooftree}[label separation=0.3em]
      \hypo{\Tm \ParallelVal \TmTwo}
      \infer1[\footnotesize$\mathsf{\lambda}$]{\lambda\Var\,\Tm \ParallelVal \lambda\Var\,\TmTwo}
    \end{prooftree}
    \quad
    \begin{prooftree}[label separation=0.3em]
      \hypo{\Tm \ParallelVal \TmTwo}
      \infer1[\footnotesize$\oc$]{\Bang{\Tm} \ParallelVal \Bang{\TmTwo}}
    \end{prooftree}
    \quad
    \begin{prooftree}[label separation=0.3em]
      \hypo{\Tm \ParallelVal \TmTwo}
      \infer1[\footnotesize$\Der{\!}$]{\Der{\Tm} \ParallelVal \Der{\TmTwo}}
    \end{prooftree}
    \quad
    \begin{prooftree}[separation=1.2em, label separation=0.3em]
      \hypo{\Tm \ParallelVal \TmTwo}
      \hypo{\TmThree \ParallelVal \TmFour}
      \infer2[\footnotesize$@$]{\App{\Tm}{\TmThree} \ParallelVal \App{\TmTwo}{\TmFour}}
    \end{prooftree}
    \quad
    \begin{prooftree}[separation=1.3em, label separation=0.3em]
      \hypo{\Tm \ParallelVal \TmTwo}
      \hypo{\TmThree \ParallelVal \TmFour}
      \infer2[\footnotesize$\ValScript$]{\App{\La{\Var}{\Tm}}{\ValThree} \ParallelVal \TmTwo\Sub{\TmFour}{\Var}}
    \end{prooftree}
  }
  \caption{Parallel $\ValScript$-reduction.}
  \label{fig:parallel}
\end{figure}

\Reflemma{development} is the key ingredient to prove the confluence of $\ToVal$ (\Reflemmap{basic-ToBang-ToVal}{ToVal-confluent} below).

The next lemma lists a series of good rewriting properties of $\ValScript$-, $\Weak{\ValScript}$-, $\Derel$- and $\Weak{\Derel}$-reductions that will be used to prove quasi-strong confluence of $\WToTot$ and confluence of $\ToTot$ (\Refprop{confluence}  below).

\newcounter{lemma:basic-ToBang-ToVal}
\addtocounter{lemma:basic-ToBang-ToVal}{\value{definition}}
\begin{lemma}[Basic properties of reductions]
\label{lemma:basic-ToBang-ToVal}\hfill
\NoteProof{lemmaAppendix:basic-ToBang-ToVal}
  \begin{enumerate}
    \item\label{lemma:basic-ToBang-ToVal.WToVal-strongly-confluent} $\WToVal$\! is quasi-strongly confluent, \Ie $\RevTo{\Weak{\ValScript}\!} \cdot \!\WToVal \, \subseteq (\WToVal \!\!\cdot \RevTo{\Weak{\ValScript}\!}) \, \cup =$\,.
    \item\label{lemma:basic-ToBang-ToVal.ToBang-strongly-confluent} $\WToBang$ and $\ToBang$ are quasi-strongly confluent (separately).
%
%
    \item \label{lemma:basic-ToBang-ToVal.commute} 
		$\WToBang$ and $\WToVal$ strongly commute (\Ie $\RevTo{\Weak{\Derel}} \cdot\! \WToVal \, \subseteq \, \WToVal \!\!\cdot \RevTo{\Weak{\Derel}} $); \,
		$\ToDer$ quasi-strongly commutes over $\ToVal$ (\Ie $\RevTo{\Derel} \,\cdot \ToVal \ \subseteq \ \ToVal \!\cdot \MRevTo{\Derel}$); \,
		$\ToBang$ and $\ToVal$ commute (\Ie $\MRevTo{\Derel} \cdot\! \ToVal^* \ \subseteq \ \ToVal^* \!\cdot \MRevTo{\Derel} $).

    \item\label{lemma:basic-ToBang-ToVal.ToVal-confluent} $\ToVal$ is confluent, \Ie $\MRevTo{\ValScript} \cdot\! \ToVal^* \ \subseteq \ \ToVal^* \!\cdot \MRevTo{\ValScript}$\,.
  \end{enumerate}
\end{lemma}

\newcounter{prop:confluence}
\addtocounter{prop:confluence}{\value{definition}}
\begin{proposition}[Quasi-strong confluence of $\WToTot$ and confluence of $\ToTot$]\hfill
\label{prop:confluence}
\NoteProof{propappendix:confluence}
  \begin{enumerate}
    \item\label{prop:confluence.WToTot} The reduction $\WToTot$ is quasi-strongly confluent, \Ie $\RevTo{\Weak{\Tot}\!} \cdot \!\WToTot \, \subseteq (\WToTot \!\!\cdot \RevTo{\Weak{\Tot}\!}) \, \cup =$\,.
    \item\label{prop:confluence.ToTot} The reduction $\ToTot$ is confluent, \Ie $\MRevTo{\Tot} \cdot\! \ToTot^* \ \subseteq \ \ToTot^* \!\cdot \MRevTo{\Tot}$\,.
  \end{enumerate}
\end{proposition}

\section{\texorpdfstring{The bang calculus with respect to \CbN\ and \CbV\ $\lambda$-calculi, syntactically}{The bang calculus with respect to \CbN\ and \CbV\ lambda-calculi, syntactically}}
\label{sect:embedding}



One of the interests of the bang calculus is that it is a general framework 
where both \emph{call-by-name} (\CbN, 
 \cite{Barendregt84}) and Plotkin's \emph{call-by-value} (\CbV, \cite{Plotkin75}) $\lambda$-calculi can be embedded.\footnotemark
\footnotetext{Here, with \CbN\ or \CbV\ $\lambda$-calculus we refer to the whole calculus and its general reduction rules, not only to \CbN\ or \CbV\ (deterministic) evaluation strategy in the $\lambda$-calculus.}
Syntax and reduction rules of \CbN\ and \CbV\ $\lambda$-calculi are in \Reffig{lambdacalculus}: \emph{$\beta$-reduction} $\ToBeta$ (resp.~\emph{$\Betav$-reduction} $\ToBetav$) is the reduction for the \CbN\ (resp.~\CbV) $\lambda$-calculus.
\CbN\ and \CbV\ $\lambda$-calculi share the same term syntax (the set $\Lambda$ of $\lambda$-terms of \CbN\ and \CbV\ $\lambda$-calculi can be seen as a subset of $\BangSet$), whereas $\ToBetav$ is just the restriction of $\ToBeta$ allowing to fire a $\beta$-redex $(\La{\Var\,}{\LTm})\LTmTwo$ only when $\LTmTwo$ is a $\lambda$-value, \Ie a variable or an abstraction.
\emph{\Ground\ $\beta$-(resp.~$\Betav$-)reduction} $\WToBeta$ (resp.~$\WToBetav$) is an interesting restriction of $\beta$-(resp.~$\Betav$-)reduction:
\begin{itemize}
  \item $\WToBeta$ is the ``hereditary'' head $\beta$-reduction, which contains head $\beta$-reduction and weak head $\beta$-reduction, two well-known evaluation strategies for \CbN\ $\lambda$-calculus (both 
  reduce the $\beta$-redex in head position, the latter does not reduce under $\lambda$'s);
  \item $\WToBetav$ is the weak $\Betav$-reduction, \Ie $\Betav$-reduction with the restriction of not reducing under $\lambda$'s; 
  it contains (weak) head $\Betav\!$-reduction (
  aka left reduction in \cite[p.\,136]{Plotkin75}), the well-known evaluation strategy for \CbV\ $\lambda$-calculus 
  firing 
  the $\Betav$-redex in left position
  (if any) not under 
  $\lambda$'s.
\end{itemize}

\begin{figure}[!t]
  \centering
  \scalebox{0.9}{\parbox{1.067\linewidth}{
  \begin{align*}
    \lambda\text{-\emph{terms}:}		&& \LTm, \LTmTwo, \LTmThree &\Coloneqq \, \LVal  \,\mid\,  \LTm\LTmTwo			&&(\textup{set: } \Lambda)  \\
    \lambda\text{-\emph{values}:}		&& \LVal &\Coloneqq \, \Var  \,\mid\, \La{\Var\,}\LTm 					&&(\textup{set: } \LambdaVal)  \\
    \lambda\text{-\emph{contexts}:}	&& \LCtx		&\Coloneqq \, \Hole{\cdot} \mid \La{\Var}{\LCtx} \mid \LCtx\LTm \mid \LTm\LCtx 	&&(\textup{set: }\LambdaCtx) \\
    \text{\emph{\CbN\ \ground\ $\lambda$-contexts}:}&&\NCtx	&\Coloneqq \, \Hole{\cdot} \mid \La{\Var}{\NCtx} \mid \NCtx\LTm 	&&(\textup{set: }\Lambda_\NCtx)	     \\
    \text{\emph{\CbV\ \ground\ $\lambda$-contexts}:}&&\VCtx	&\Coloneqq \, \Hole{\cdot} \mid \VCtx\LTm \mid \LTm\VCtx  		&&(\textup{set: }\Lambda_\VCtx)	     \\[-2\baselineskip]
  \end{align*}
  \begin{align*}
    \text{\emph{Root-steps}:}	&& (\La{\Var}{\LTm}){\TmTwo} &\Root{\beta} \LTm\Sub{\TmTwo}{\Var} \ (\text{\CbN})  \qquad \ (\La{\Var}{\LTm}){\LVal} \Root{\Betav} \LTm\Sub{\LVal}{\Var} \ (\text{\CbV}) \\[.2\baselineskip]
    \Rule\text{-\emph{reduction}:} &&   \LTm \To{\Rule}  \LTmTwo \ &\Leftrightarrow \ \exists \, \Ctx \in \Lambda_\Ctx, \, \exists\, \LTm'\!, \LTmTwo' \!\in \Lambda : \LTm = \Ctxp{\LTm'}, \, \LTmTwo = \Ctxp{\LTmTwo'}, \, \LTm' \!\Root{\Rule} \LTmTwo' \\
    \Weak{\beta}\text{-\emph{reduction}:} &&   \LTm \WToBeta  \LTmTwo \ &\Leftrightarrow \ \exists \, \NCtx \in \Lambda_\NCtx, \,\exists\, \LTm'\!, \LTmTwo' \!\in \Lambda : \LTm = \NCtxp{\LTm'}, \, \LTmTwo = \NCtxp{\LTmTwo'}, \, \LTm' \!\Root{\beta} \LTmTwo' \\
    \Weak{\Betav}\text{-\emph{reduction}:} &&   \LTm \WToBetav  \LTmTwo \ &\Leftrightarrow \ \exists \, \VCtx \in \Lambda_\VCtx, \,\exists\, \LTm'\!, \LTmTwo' \!\in \Lambda : \LTm = \VCtxp{\LTm'}, \, \LTmTwo = \VCtxp{\LTmTwo'}, \, \LTm' \!\Root{\Betav} \LTmTwo'
    \\[-1.5\baselineskip]
  \end{align*}
  }}
  \caption{The \CbN\ and \CbV\ $\lambda$-calculi: their syntax and reduction rules, where $\Rule \in \{\beta, \Betav\}$.}
  \label{fig-lambdacalculus}
\end{figure}

\paragraph{\CbN\ and \CbV\ translations into the bang calculus.}
The \emph{\CbN} and \emph{\CbV} \emph{translations} are two functions $\Cbn{(\cdot)} \colon \Lambda \to \BangSet$ and $\Cbv{(\cdot)} \colon \Lambda \to \BangSet$, respectively, translating $\lambda$-terms into terms of the bang calculus: 
\begin{align*}
  \Cbn{\Var} &\Defeq \Var 						& \ \Cbn{(\La{\Var\,}{\LTm})} &\Defeq \La{\Var\,}{\Cbn{\LTm}} & \ \Cbn{(\LTm\LTmTwo)} &\Defeq \App{\Cbn{\LTm}}{\Bang{{\Cbn{\LTmTwo}}}} \,; \\ 
  \Cbv{\Var} &\Defeq \Bang{\Var} & \	\Cbv{(\La{\Var\,}{\LTm})} &\Defeq \Bang{(\La{\Var\,}{\Cbv{\LTm}})} 
	& \ \Cbv{(\LTm\LTmTwo)} &\Defeq \App{\Der{\,\Cbv{\LTm}}}{\Cbv{\LTmTwo}} \,.
\end{align*}

\begin{example}
\label{ex:delta-translated}
  Let $\omega \Defeq (\lambda \Var \, \Var\Var)\lambda \Var \, \Var\Var$, the typical diverging $\lambda$-term for \CbN\ and \CbV\ $\lambda$-calculi: one has 
  $\Cbn{\omega} = \App{\Delta}{\Bang{\Delta}}$ and $\Cbv{\omega} = \App{\Derp{{\Bang{\Delta\!'}}}}{\Bang{\Delta\!'}} 
  $, 
  which are not $\Weak{\Tot}$- nor $\Tot$-normalizable 
  ($\Delta$ and $\Delta\!'$ are defined in \Refex{delta}). 
\end{example}

For any $\lambda$-term $\LTm$, $\Cbn{\LTm}$ and $\Cbv{\LTm}$ are just different decorations of $\LTm$ by means of the monadic operators $\oc$ and $\Der{\!}$ (the latter does not occur in $\Cbn{\LTm}$).
Note that the translation $\Cbn{(\cdot)}$ puts the argument of any application into a box: in \CbN\ $\lambda$-calculus any $\lambda$-term is duplicable or discardable.
On the other hand, only $\lambda$-values (\Ie abstractions and variables) are translated by $\Cbv{(\cdot)}$ into boxes, as they are the only $\lambda$-terms duplicable or discardable in \CbV\ $\lambda$-calculus.

\newcounter{lemma:substitution}
\addtocounter{lemma:substitution}{\value{definition}}
\begin{lemma}[Substitution]
  \label{lemma:substitution}
  Let 
  $\LTm, \LTmTwo$ be $\lambda$-terms and $\Var$ be a variable.
  \begin{enumerate}
    \item\label{lemma:substitution.cbn}\emph{\CbN\ translation vs.~substitution:} One has that $\Cbn{\LTm} \Sub{\Cbn{\LTmTwo}\!}{\Var} = \Cbn{(\LTm\Sub{\LTmTwo}{\Var})}$.
    \item\label{lemma:substitution.cbv}\emph{\CbV\ translation vs.~substitution:} If $\LTmTwo$ is such that $\Cbv{\LTmTwo} = \Bang{\TmTwo}$ for some $\TmTwo \in \BangSet$, then $\Cbv{\LTm} \Sub{\TmTwo}{\Var} = \Cbv{(\LTm\Sub{\LTmTwo}{\Var})}$.
  \end{enumerate}
\end{lemma}

\begin{proof}   The proofs of both points are by induction on 
	$\LTm \in \Lambda$.
	\begin{enumerate}
		\item 
		\begin{itemize}
			\item \emph{Variable:} If $\LTm$ is a variable then there are two subcases.
			If $\LTm \Defeq \Var$ then 
			$\Cbn{\LTm} = {\Var}$, so $\Cbn{\LTm} \Sub{\Cbn{\LTmTwo}}{\Var} = \Cbn{\LTmTwo} = \Cbn{(\LTm\Sub{\LTmTwo}{\Var})}$.
			Otherwise $\LTm \Defeq \VarTwo \neq \Var$ and then $\Cbn{\LTm} = \VarTwo$, hence $\Cbn{\LTm} \Sub{\Cbn{\LTmTwo}}{\Var} = \VarTwo = \Cbn{(\LTm\Sub{\LTmTwo}{\Var})}$.
			
			\item \emph{Abstraction:} If $\LTm \Defeq \La{\VarTwo}{\LTmThree}$ then $\Cbn{\LTm} = \La{\VarTwo}{\Cbn{\LTmThree}}$ 
			(we 
			suppose without loss of generality 
			$\VarTwo \notin \Fv{\LTmTwo} \cup \{\Var\}$). 
			By 
			\Ih, $\Cbn{\LTmThree} \Sub{\Cbn{\LTmTwo}\!}{\Var} = \Cbn{(\LTmThree \Sub{\LTmTwo}{\Var})}$ and so $\Cbn{\LTm} \Sub{\Cbn{\LTmTwo}\!}{\Var} = \La{\VarTwo}{(\Cbn{\LTmThree} \Sub{\Cbn{\LTmTwo}\!}{\Var})} = \La{\VarTwo}{\Cbn{(\LTmThree \Sub{\LTmTwo}{\Var})}} = \Cbn{(\LTm\Sub{\LTmTwo}{\Var})}$.
			
			\item \emph{Application:} If $\LTm \Defeq \LTmThree\LTmFour$ then $\Cbn{\LTm} = \App{\Cbn{\LTmThree}}{\Bang{\Cbn{\LTmFour}}}$.
			By 
			\Ih, $\Cbn{\LTmThree} \Sub{\Cbn{\LTmTwo}}{\Var} = \Cbn{(\LTmThree \Sub{\LTmTwo}{\Var})}$ and $\Cbn{\LTmFour} \Sub{\Cbn{\LTmTwo}}{\Var} \allowbreak= \Cbn{(\LTmFour \Sub{\LTmTwo}{\Var})}$.
			So, 
			$
			\Cbn{\LTm} \Sub{\Cbn{\LTmTwo}}{\Var} = \App{\Cbn{\LTmThree} \Sub{\Cbn{\LTmTwo}}{\Var}}{\Bang{(\Cbn{\LTmFour} \Sub{\Cbn{\LTmTwo}}{\Var})}} = \App{\Cbn{(\LTmThree\Sub{\LTmTwo}{\Var})}}{\Bang{({\Cbn{(\LTmFour\Sub{\LTmTwo}{\Var})}} )}} \allowbreak= \Cbn{(\LTm\Sub{\LTmTwo}{\Var})}$.
		\end{itemize}
		
		\smallskip
		\item 
		\begin{itemize}
			\item \emph{Variable:} If $\LTm$ is a variable then there are two subcases.
			If $\LTm \Defeq \Var$ then 
			$\Cbv{\LTm} = \Bang{\Var}$, so $\Cbv{\LTm} \Sub{\TmTwo}{\Var} = \Bang{\TmTwo} = \Cbv{\LTmTwo} = \Cbv{(\LTm\Sub{\LTmTwo}{\Var})}$.
			Otherwise $\LTm \Defeq \VarTwo \neq \Var$ and then $\Cbv{\LTm} = \Bang{\VarTwo}$, hence $\Cbv{\LTm} \Sub{\TmTwo}{\Var} = \Bang{\VarTwo} = \Cbv{(\LTm\Sub{\LTmTwo}{\Var})}$.
			
			\item \emph{Application:} If $\LTm \Defeq \LTmFive\LTmFour$ then $\Cbv{\LTm} = \App{\Der{\Cbv{\LTmFive}}}{\Cbv{\LTmFour}}$\!.
			By 
			\Ih, $\Cbv{\LTmFive} \Sub{\TmTwo}{\Var} = \Cbv{(\LTmFive \Sub{\LTmTwo}{\Var})}$ and $\Cbv{\LTmFour} \Sub{\TmTwo}{\Var} \allowbreak= \Cbv{(\LTmFour \Sub{\LTmTwo}{\Var})}$.
			So, $\Cbv{\LTm} \Sub{\TmTwo}{\Var} \allowbreak= \App{ \Derp{\Cbv{\LTmFive} \Sub{\TmTwo}{\Var}}}{\Cbv{\LTmFour} \Sub{\TmTwo}{\Var}} =\allowbreak \App{ \Derp{\Cbv{\LTmFive\Sub{\LTmTwo}{\Var}}}}\allowbreak{\Cbv{(\LTmFour\Sub{\LTmTwo}{\Var})}} \allowbreak= \Cbv{(\LTm\Sub{\LTmTwo}{\Var})}$.
									
			\item \emph{Abstraction:} If $\LTm \Defeq \La{\!\VarTwo}{\LTmThree}$ then $\Cbv{\LTm} \!=\! \Bang{(\La{\!\VarTwo}{\Cbv{\LTmThree}})}$ 
			(suppose without loss of generality 
			$\VarTwo \notin \Fv{\LTmTwo} \cup \{\Var\}$). 
			By 
			\Ih\ $\Cbv{\LTmThree}\! \Sub{\TmTwo}{\Var} = \Cbv{(\LTmThree \Sub{\LTmTwo}{\Var})}$\!, 
			so $\Cbv{\LTm} \!\Sub{\TmTwo}{\Var} = \Bangp{\La{\!\VarTwo}{(\Cbv{\LTmThree} \!\Sub{\TmTwo}{\Var})}} = \Bangp{\La{\!\VarTwo}{\Cbv{(\LTmThree \Sub{\LTmTwo}{\Var})}}} = \Cbv{(\LTm\Sub{\LTmTwo}{\Var})}$\!.
			\qedhere
		\end{itemize}
	\end{enumerate}
\end{proof}

Note that the hypothesis about $\LTmTwo$ in \Reflemmap{substitution}{cbv} is fulfilled 
if and only if $\LTmTwo$ is a $\lambda$-value.

\begin{remark}[\CbV\ translation is $\ValScript$-normal]
	\label{rmk:cbv-val-normal}
	It is immediate to prove by induction on 
	$\LTm \in \Lambda$ that $\Cbv{\LTm}$ is $\ValScript$-normal, so if $\Cbv{\LTm} \ToBang \TmTwo_0 \ToVal \TmTwo$ then the only $\ValScript$-redex in $\TmTwo_0$ has been created by 
	the step $\Cbv{\LTm} \ToBang \TmTwo_0$ and is absent in $\Cbv{\LTm}$.
\end{remark}

\paragraph{Simulating  \CbN\ and \CbV\ reductions into the bang calculus.}
We can now show that the \CbN\ translation $\Cbn{(\cdot)}$ (\Resp \CbV\ translation $\Cbv{(\cdot)}$) from the \CbN\ (\Resp \CbV) $\lambda$-calculus into the bang calculus is \emph{sound} and \emph{complete}: it maps $\beta$-reductions (\Resp $\Betav$-reductions) of the $\lambda$-calculus into $\Tot$-reductions of the bang calculus, and conversely $\Tot$-reductions\,---\,when  restricted to the image of the translation\,---\,into $\beta$-reductions (\Resp $\Betav$-reductions). 
Said differently, the target of the \CbN\ (\Resp \CbV) translation into the bang calculus is a \emph{conservative} extension of the \CbN\ (\Resp \CbV) $\lambda$-calculus.

\newcounter{thm:embedding}
\addtocounter{thm:embedding}{\value{definition}}
\begin{theorem}[Simulation of \CbN\ and \CbV\ $\lambda$-calculi]
\label{thm:embedding}
  Let
  $\LTm$ be a $\lambda$-term. 
  \begin{enumerate}
    \item\label{thm:embedding.cbn} \emph{Conservative extension of \CbN\ $\lambda$-calculus:}\\
    \emph{Soundness:} If $\LTm \ToBeta \LTm'$ then $\Cbn{\LTm} \ToVal \Cbn{{\LTm'}}$ (and $\Cbn{\LTm} \ToTot \Cbn{{\LTm'}}$); \\
    \emph{Completeness:} Conversely, if $\Cbn{\LTm} \ToTot \TmTwo$ then $\Cbn{\LTm} \ToVal \TmTwo = \Cbn{{\LTm'}}$ and $\LTm \ToBeta \LTm'$ for some $\lambda$-term $\LTm'$.
    \item\label{thm:embedding.cbn-head} \emph{Conservative extension of \ground\ \CbN\ $\lambda$-calculus:}\\
    \emph{Soundness:} If $\LTm \WToBeta \LTm'$ then $\Cbn{\LTm} \WToVal \Cbn{{\LTm'}}$ (and $\Cbn{\LTm} \WToTot \Cbn{{\LTm'}}$);\\
    \emph{Completeness:} Conversely, if $\Cbn{\LTm} \WToTot \TmTwo$ then $\Cbn{\LTm} \WToVal \TmTwo = \Cbn{{\LTm'}}$ and \mbox{$\LTm \WToBeta \LTm'$} for some $\lambda$-term $\LTm'$.
    \item\label{thm:embedding.cbv}
    \emph{Conservative extension of \CbV\ $\lambda$-calculus:}\\
    \emph{Soundness:} If $\LTm \To{\Betav} \LTm'$ then $\Cbv{\LTm} \ToBang\ToVal \Cbv{{\LTm'}}$ (and hence $\Cbv{\LTm} \ToTot\ToTot \Cbv{{\LTm'}}$);\\
    \emph{Completeness:} Conversely, if $\Cbv{\LTm} \ToBang\ToVal \TmTwo$ then $\TmTwo = \Cbv{{\LTm'}}$ and $\LTm \ToBetav \LTm'$ for some $\lambda$-term $\LTm'$.
    \item\label{thm:embedding.cbv-head}
    \emph{Conservative extension of \ground\ \CbV\ $\lambda$-calculus:} \\
    \emph{Soundness:} If $\LTm \WToBetav \LTm'$ then $\Cbv{\LTm} \WToBang\WToVal \Cbv{{\LTm'}}$ (and hence $\Cbv{\LTm} \WToTot\WToTot \Cbv{{\LTm'}}$);\\
    \emph{Completeness:} Conversely, if $\Cbv{\LTm} \WToBang\WToVal \TmTwo$ then $\TmTwo = \Cbv{{\LTm'}}$ and $\LTm \WToBetav \LTm'$ for some $\lambda$-term $\LTm'$.
  \end{enumerate}
\end{theorem}

\begin{proof}
	\begin{enumerate}
		\item \emph{Soundness:} We prove by induction on the $\lambda$-term $\LTm$ that if $\LTm \ToBeta \LTm'$ then $\Cbn{\LTm} \ToVal \Cbn{{\LTm'}}$ (this implies $\Cbn{\LTm} \ToTot \Cbn{{\LTm'}}$, since $\ToVal \,\subseteq\, \ToTot$). 
		According to the definition of $\LTm \ToBeta \LTm'$, there are the following cases:
		\begin{itemize}
			\item \emph{Root-step}, \Ie $\LTm \Defeq (\La{\Var}{\LTmThree})\LTmTwo \Root{\beta} \LTmThree\Sub{\LTmTwo}{\Var} \Eqdef {\LTm'}$: by \Reflemmap{substitution}{cbn}, 
			$\Cbn{\LTm} = \App{\La{\Var}{\Cbn{\LTmThree}}}{\Bang{\Cbn{\LTmTwo}}} \allowbreak\Root{\ValScript} \Cbn{\LTmThree} \Sub{\Cbn{\LTmTwo}\!}{\Var} = \Cbn{{\LTm'}}$.
			
			\item  \emph{Abstraction}, \Ie $\LTm \Defeq \La{\Var}{\LTmThree} \ToBeta \La{\Var}{\LTmThree'} \Eqdef {\LTm'}$ with $\LTmThree \ToBeta \LTmThree'$: by 
			\Ih, $\Cbn{\LTmThree} \ToVal \Cbn{{\LTmThree'}}$, thus $\Cbn{\LTm} = \La{\Var}{\Cbn{\LTmThree}} \ToVal \La{\Var}{\Cbn{{\LTmThree'}}} = \Cbn{{\LTm'}}$.
			
			\item \emph{Application left}, \Ie $\LTm \Defeq \LTmThree\LTmTwo \ToBeta \LTmThree'\LTmTwo \Eqdef {\LTm'}$ with $\LTmThree \ToBeta \LTmThree'$:
			analogous to the previous case.
			
			\item \emph{Application right}, \Ie $\LTm \Defeq \LTmTwo\LTmThree \ToBeta \LTmTwo\LTmThree' \!\Eqdef \LTm'$ with $\LTmThree \ToBeta\! \LTmThree'$: 
			by 
			\Ih $\Cbn{\LTmThree} \ToVal \Cbn{{\LTmThree'}}$\!, so $\Cbn{\LTm} = \App{\Cbn{\LTmTwo}}{\Bang{{\Cbn{\LTmThree}}}} \!\ToVal \App {\Cbn{\LTmTwo}} {\Bang{{\Cbn{{\LTmThree'}}}}} \!= \Cbn{{\LTm'}}$\!.
		\end{itemize}
		
		\emph{Completeness:} First, observe that $\Cbn{\LTm} \ToTot \TmTwo$ entails $\Cbn{\LTm} \ToVal \TmTwo$ since $\Der{\!}$ does not occur in $\Cbn{\LTm}$, hence $\Cbn{\LTm}$ is $\Derel$-normal.
		We prove by induction on the $\lambda$-term $\LTm$ that if $\Cbn{\LTm} \ToVal \TmTwo$ then $\TmTwo = \Cbn{{\LTm'}}$ and $\LTm \ToBeta \LTm'$ for some $\lambda$-term $\LTm'$. 
		According to the definition of $\Cbn{\LTm} \ToVal \TmTwo$, there are the following cases:
		\begin{itemize}
			\item \emph{Root-step}, \Ie $\Cbn{\LTm} \Defeq \App{\La{\Var}{\Cbn{\LTmThree}}}{\Bang{\Cbn{\LTmFour}}} \allowbreak\Root{\ValScript} \Cbn{\LTmThree} \Sub{\Cbn{\LTmFour}}{\Var} \Eqdef \TmTwo$: 
			by \Reflemmap{substitution}{cbn} $\TmTwo = \Cbn{(\LTmThree \Sub{\LTmFour}{\Var})}$, so $\LTm = (\La{\Var}{\LTmThree})\LTmFour \Root{\beta} \LTmThree\Sub{\LTmFour}{\Var} \Eqdef {\LTm'}$ where $\Cbn{{\LTm'}} = \TmTwo$.
			
			\item  \emph{Abstraction}, \Ie $\Cbn{\LTm} \Defeq \La{\Var}{\Cbn{\LTmThree}} \ToVal \La{\Var}{\TmTwo'} \Eqdef \TmTwo$ with $\Cbn{\LTmThree} \ToVal \TmTwo'$: by 
			\Ih, there is a $\lambda$-term $\LTmThree'$ such that $\Cbn{{\LTmThree'}} = \TmTwo'$ and $\LTmThree \ToBeta \LTmThree'$, thus $\LTm = \La{\Var}{\LTmThree} \ToBeta \La{\Var}{\LTmThree'} \Eqdef {\LTm'}$ where $\Cbn{{\LTm'}} = \La{\Var}\Cbn{{\LTmThree'}} = \TmTwo$.
			
			\item \emph{Application left}, \Ie $\Cbn{\LTm} \Defeq \App{\Cbn{\LTmThree}}{\Bang{{\Cbn{\LTmFour}}}} \ToVal \App{\TmTwo'}{\Bang{{\Cbn{\LTmFour}}}} \Eqdef \TmTwo$ with $\Cbn{\LTmThree} \ToVal \TmTwo'$: 
			analogously to above.
			
			\item \emph{Application right}, \Ie $\Cbn{\LTm} \!\Defeq \App{\Cbn{\LTmFour}}{\Bang{{\Cbn{\LTmThree}}}} \!\ToVal \App {\Cbn{\LTmFour}} \Bang{\TmTwo'} \!\Eqdef \TmTwo$ with $\Cbn{\LTmThree} \ToVal \TmTwo'$: by 
			\Ih, there is a $\lambda$-term $\LTmThree'$ such that $\Cbn{{\LTmThree'}} = \TmTwo'$ and $\LTmThree \ToBeta \LTmThree'$, so $\LTm = \LTmFour\LTmThree \ToBeta \LTmFour\LTmThree' \Eqdef \LTm'$ with $\Cbn{{\LTm'}} = \App{\Cbn{\LTmFour}}{\Bang{\Cbn{{\LTmThree'}}}} = \TmTwo$.
		\end{itemize}
		
		\item Since $\Root{\beta}$ is simulated by $\Root{\ValScript}$ and vice-versa (see the root-cases 
		above), 
		\Refthmp{embedding}{cbn-head} is proved analogously to the proof of \Refthmp{embedding}{cbn} ($\WToVal$ replaces $\ToVal$, and $\WToBeta$ replaces $\ToBeta$), with the difference that, by definition, $\WToBeta$ and $\WToVal$ do not give rise to the case \emph{Application right}.
		
		\item \emph{Soundness:} We prove by induction on the $\lambda$-term $\LTm$ that if $\LTm \ToBetav \LTm'$ then $\Cbv{\LTm} \ToDer\ToVal \Cbv{{\LTm'}}$.
		According to the definition of $\LTm \ToBetav \LTm'$, there are the following cases:
		\begin{itemize}
			\item \emph{Root-step}, \Ie $\LTm \Defeq (\La{\Var}{\LTmThree})\LVal \Root{\Betav} \LTmThree\Sub{\LVal}{\Var} \Eqdef \LTm'$ where $\LVal$ is a $\lambda$-value, \Ie a variable or an abstraction: then $\Cbv{\LVal} = \Bang{\TmTwo}$ for some $\TmTwo \in \BangSet$, hence $\Cbv{\LTm} = \App{\Der{\Bang{(\La{\Var}{\Cbv{\LTmThree}})}}}{\Cbv{\LVal}} \allowbreak\WToDer\allowbreak \App{\La{\Var}{\Cbv{\LTmThree}}}{\Cbv{\LVal}} \Root{\ValScript}\allowbreak \Cbv{\LTmThree} \Sub{\TmTwo}{\Var} = \Cbv{{\LTm'}}$ by \Reflemmap{substitution}{cbv} (recall that $\WToDer \,\subseteq\,\ToDer$).
			
			\item  \emph{Abstraction}, \Ie $\LTm \Defeq \La{\Var}{\LTmThree} \ToBetav \La{\Var}{\LTmThree'} \Eqdef \LTm'$ with $\LTmThree \ToBetav \LTmThree'$: by 
			\Ih, $\Cbv{\LTmThree} \ToDer\ToVal \Cbv{{\LTmThree'}}$; therefore $\Cbv{\LTm} = \Bang{(\La{\Var}{\Cbv{\LTmThree}})} \ToDer\ToVal \Bang{(\La{\Var}{\Cbv{{\LTmThree'}}})} = \Cbv{{\LTm'}}$.
			
			\item \emph{Application left}, \Ie $\LTm \Defeq \LTmThree\LTmTwo \ToBetav \LTmThree'\LTmTwo \Eqdef \LTm'$ with $\LTmThree \ToBetav \LTmThree'$: by 
			\Ih $\Cbv{\LTmThree} \ToDer\ToVal \Cbv{{\LTmThree'}}$; therefore $\Cbv{\LTm} = \App{\Der{\Cbv{\LTmThree}}}{{\Cbv{\LTmTwo}}} \ToDer\ToVal \App{\Der{\Cbv{{\LTmThree'}}}}{{\Cbv{\LTmTwo}}} = \Cbv{{\Tm'}}$.
			
			\item \emph{Application right}, \Ie $\LTm \Defeq \LTmTwo\LTmThree \ToBetav \LTmTwo\LTmThree' \!\Eqdef \LTm'$ with $\LTmThree \ToBetav\! \LTmThree'$: 
			analogous to the previous case.
		\end{itemize}
		
		\emph{Completeness:} 
		We prove by induction on $\TmTwo_0 \in \BangSet$ that if $\Cbv{\LTm} \ToDer \TmTwo_0 \ToVal \TmTwo$ then $\TmTwo = \Cbv{{\LTm'}}$ and $\LTm \ToBetav \LTm'$ for some $\lambda$-term $\LTm'$.
		According to the definition of $\TmTwo_0 \ToVal \TmTwo$, there are the following cases:
		\begin{itemize}
			\item \emph{Root-step}, \Ie $\TmTwo_0 \Defeq \App{\La{\Var}{\TmThree}}{\Bang{\TmFour}} \Root{\ValScript} \TmThree\Sub{\TmFour}{\Var} \Eqdef \TmTwo$.
			By \Refrmk{cbv-val-normal}, necessarily $\Cbv{\LTm} = \App{\Der{\Bang{(\La{\Var}{\TmThree})}}}{\Bang{\TmFour}}$ and hence $\LTm = (\La{\Var}{\LTmThree_0})\LVal$ for some $\lambda$-term $\LTmThree_0$ and some $\lambda$-value $\LVal$ such that $\Cbv{\LTmThree_0} = \TmThree$ and $\Cbv{\LVal} = \Bang{\TmFour}$.
			Note that $\Cbv{\LTm} \WToDer \TmTwo_0$. 
			Let $\LTm' \Defeq \LTmThree_0 \Sub{\LVal}{\Var}$: then, $\LTm \Root{\Betav} \LTm'$ and $\Cbv{{\LTm'}} = \Cbv{\LTmThree_0} \Sub{\TmFour}{\Var} = \TmTwo$ by \Reflemmap{substitution}{cbv}.
			
			\item \emph{Abstraction}, \Ie $\TmTwo_0 \Defeq \La{\Var}{\TmThree_0} \ToVal \La{\Var}{\TmThree'} \Eqdef \TmTwo$ with $\TmThree_0 \ToVal \TmThree'$.
			This case is impossible because, according to \Refrmk{cbv-val-normal}, necessarily $\Cbv{\LTm} = \La{\Var}{\TmThree}$ for some $\ValScript$-normal $\TmThree \in \BangSet$ such that $\TmThree \ToDer \TmThree_0$, but there is no $\lambda$-term $\LTm$ such that $\Cbv{\LTm}$ is an abstraction.
			\item \emph{Dereliction}, \Ie $\TmTwo_0 \Defeq \Der{\TmThree_0} \ToVal \Der{\TmThree'} \Eqdef \TmTwo$ with $\TmThree_0 \ToVal \TmThree'$.
			This case is impossible because, according to \Refrmk{cbv-val-normal}, necessarily $\Cbv{\LTm} = \Der{\TmThree}$  for some $\ValScript$-normal $\TmThree \in \BangSet$ such that $\TmThree \ToDer \TmThree_0$, but there is no $\lambda$-term $\LTm$ such that $\Cbv{\LTm}$ is a dereliction.
			
			\item \emph{Application left}, \Ie $\TmTwo_0 \Defeq \App{\TmThree_0}{\TmThree_1} \ToVal \App{\TmThree'}{\TmThree_1} \Eqdef \TmTwo$ with $\TmThree_0 \ToVal \TmThree'$.
			By \Refrmk{cbv-val-normal}, 
			$\Cbv{\LTm} = \App{\Der{\TmFive}}{\TmThree_1}$ for some $\ValScript$-normal $\TmFive \in \BangSet$ such that $\Der{\TmFive} \ToDer \TmThree_0$, and thus $\TmThree_0 = \Der{\TmFive_0}$ where $\TmFive \ToDer \TmFive_0$ (indeed $\TmFive = \Bang{\TmThree_0}$ is impossible because $\TmFive$ is $\ValScript$-normal), and hence $\TmThree' = \Der{\TmFive'}$ with $\TmFive_0 \ToVal \TmFive'$.
			So, $\LTm = \LTmFour\LTmFour_1$ for some $\lambda$-terms $\LTmFour$ and $\LTmFour_1$ such that $\Cbv{\LTmFour} = \TmFive$ and $\Cbv{\LTmFour_1} = \TmThree_1$ with $\Cbv{\LTmFour} \ToDer \TmFive_0 \ToVal \TmFive'$.
			By 
			\Ih, $\TmFive' = \Cbv{{\LTmFour'}}$ and $\LTmFour \ToBetav \LTmFour'$ for some $\lambda$-term $\LTmFour'$.
			Let $\LTm' \Defeq \LTmFour'\LTmFour_1$: then, $\LTm = \LTmFour\LTmFour_1 \ToBetav 
			\LTm'$ and $\Cbv{{\LTm'}} = \App{\Der{\Cbv{{\LTmFour'}}}}{\Cbv{\LTmFour_1}} = \TmTwo$.
			
			\item \emph{Application right}, \Ie $\TmTwo_0 \Defeq \App{\TmThree_1}{\TmThree_0} \ToVal \App{\TmThree_1}{\TmThree'} \Eqdef \TmTwo$ with $\TmThree_0 \ToVal \TmThree'$.
			By \Refrmk{cbv-val-normal}, necessarily $\Cbv{\LTm} = \App{\Der{\TmFive_1}}{\TmThree}$ for some $\ValScript$-normal $\TmThree, \TmFive_1 \in \BangSet$ such that $\Der{\TmFive_1} = \TmThree_1$ and $\TmThree \ToDer \TmThree_0$.
			So, $\LTm = \LTmFour_1\LTmFour$ for some $\lambda$-terms $\LTmFour_1$ and $\LTmFour$ such that $\Cbv{\LTmFour_1} = \TmFive_1$ and $\Cbv{\LTmFour} = \TmThree$ with $\Cbv{\LTmFour} \ToDer \TmThree_0 \ToVal \TmThree'$.
			By 
			\Ih, $\TmThree' = \Cbv{{\LTmFour'}}$ and $\LTmFour \ToBetav \LTmFour'$ for some $\lambda$-term $\LTmFour'$.
			Let $\LTm' \Defeq \LTmFour_1\LTmFour'$: then, $\LTm = \LTmFour_1\LTmFour \ToBetav 
			\LTm'$ and $\Cbv{{\LTm'}} = \App{\Der{\Cbv{\LTmFour_1}}}{\Cbv{{\LTmFour'}}} = \TmTwo$.
			
			\item \emph{Box}, \Ie $\TmTwo_0 \Defeq \Bang{\TmThree_0} \ToVal \Bang{\TmThree'} \Eqdef \TmTwo$ with $\TmThree_0 \ToVal \TmThree'$.
			According to \Refrmk{cbv-val-normal}, necessarily $\Cbv{\LTm} = \Bang{\TmThree}$ for some $\ValScript$-normal $\TmThree \in \BangSet$ such that $\TmThree \ToDer \TmThree_0$.
			So, $\LTm = \La{\Var}{\LTmFour}$ (since $\Cbv{\LTm}$ is a box and $\Cbv{\Var}$ is $\Derel$-normal) for some $\lambda$-term $\LTmFour$ such that $\TmThree = \La{\Var}\Cbv{\LTmFour}$, and hence there are $\TmFive_0, \TmFive' \in \BangSet$ such that $\TmThree_0 = \La{\Var}{\TmFive_0}$ and $\TmThree' = \La{\Var}{\TmFive'}$  with $\Cbv{\LTmFour} \ToDer \TmFive_0 \ToVal \TmFive'$.
			By 
			\Ih, $\TmFive' = \Cbv{{\LTmFour'}}$ and $\LTmFour \ToBetav \LTmFour'$ for some $\lambda$-term $\LTmFour'$.
			Let $\LTm' \Defeq \La{\Var}{\LTmFour'}$: then, $\LTm = \La{\Var}{\LTmFour} \ToBetav 
			\LTm'$ and $\Cbv{{\LTm'}} = \Bangp{\La{\Var}{\Cbv{{\LTmFour'}}}} = \TmTwo$.
		\end{itemize}
		
		\item Since $\Root{\Betav}$ is simulated by $\WToDer\Root{\ValScript}$ and vice-versa (see the root-cases 
		above), \Refthmp{embedding}{cbv-head} is proved analogously to the proof of \Refthmp{embedding}{cbv} (replace $\ToVal$ with $\WToVal$, and $\ToDer$ with $\WToDer$, as well as $\ToBetav$ with $\WToBetav$), with the difference that $\WToBetav$ does not give rise to the case \emph{Abstraction} (in the soundness proof) and \emph{Box} (in the completeness proof). 
		\qedhere
	\end{enumerate}
\end{proof}

So, 
the bang calculus can simulate $\beta$- and $\Betav$-reductions via $\Cbn{(\cdot)}$ and $\Cbv{(\cdot)}$ and, conversely, $\ValScript$-reductions in the targets of $\Cbn{(\cdot)}$ and $\Cbv{(\cdot)}$ correspond to $\beta$- and $\Betav$-reductions
. Also, these simulations are:
\begin{itemize}
  \item \emph{modular}, in the sense that \emph{\ground} $\beta$-reduction (including head $\beta$-reduction and weak head $\beta$-reduction) is simulated by \emph{\ground} $\ValScript$-reduction, and vice-versa (\Refthmp{embedding}{cbn-head}); 
  \emph{\ground} $\Betav\!$-reduction (including 
  head $\Betav\!$-reduction) is simulated by \emph{\ground} $\Derel$- and $\ValScript$-reductions, and vice-versa (\Refthmp{embedding}{cbv-head});
  
  \item \emph{quantitative sensitive}, meaning that \emph{one step} of (\ground) $\beta$-reduction corresponds exactly, via $\Cbn{(\cdot)}$, to \emph{one step} of (\ground) $\ValScript$-reduction, and vice-versa;
  \emph{one step} of (\ground) $\Betav$-reduction corresponds exactly, via $\Cbv{(\cdot)}$, to \emph{one step} of (\ground) $\ValScript$-reduction, and vice-versa.
\end{itemize}
The target of \CbN\ translation $\Cbn{(\cdot)}$  
into the bang calculus can be \emph{characterized syntactically} (\Refrmk{cbn-image}).

\begin{remark}[Image of \CbN\ translation]
	\label{rmk:cbn-image}
	The \CbN\ translation $\Cbn{(\cdot)}$ 
	is a \emph{bijection} from the set $\Lambda$ of $\lambda$-terms to the subset $\Cbn{\BangSet\!}$ of $\BangSet$ defined in \Reffig{targets}:
	\begin{figure}[!t]
  \centering
  \scalebox{0.9}{\parbox{1.08\linewidth}{
  \begin{align*}
    \text{\emph{Target of \CbN\ translation into} $\BangSet$:}		&& \Tm, \TmTwo &\Coloneqq \Var \mid \App{\Tm}{\Bang{\TmTwo}} \mid \La{\Var}{\Tm}			&&\textup{\!(set: } \Cbn{\BangSet\!}\textup{)}  \\
    \text{\emph{Target of \CbV\ translation into} $\BangSet$:}		&& \ImTmOne, \ImTmTwo &\Coloneqq \Bang{\ImVal} \mid \App{\Der{\ImTmOne}}{\ImTmTwo} \mid \App{\ImVal}{\ImTmOne} 					&&\textup{\!(set: }\Cbv{\BangSet\!}\textup{)} &\qquad\!\!\! \ImVal &\Coloneqq \Var \mid \La{\Var}{\ImTmOne} &&\textup{\!(set: }\Cbv{\BangSet_v}\textup{).}
  \end{align*}
  
  \vspace{-\baselineskip}
  }}
  \caption{Targets of \CbN\ and \CbV\ translations into the bang calculus.}
  \label{fig-targets}
\end{figure}

	$\Cbn{\LTm} \in \Cbn{\BangSet\!}$ for any $\LTm \in \Lambda$, and conversely,
	for any $\Tm \in \Cbn{\BangSet\!}$, there is a \emph{unique} $\LTm \in \Lambda$ such that $\Cbn{\Tm} = \LTm$. 
	According to the definition of $\Cbn{\BangSet\!}$ (\Reffig{targets}), the construct $\Der{\!}$ never occurs in any term in $\Cbn{\BangSet\!}$, hence 
	the relations $\ToDer$ and $\WToDer$ are empty and 
	$\ToVal \ =\ \ToTot$ and $\WToVal \,=\ \WToTot$ in $\Cbn{\BangSet\!}$.
\end{remark}

\Refthmps{embedding}{cbn}{cbn-head} and \Refrmk{cbn-image} mean that $\Cbn{\BangSet\!}$ endowed with the reduction $\ToVal$ (\Resp $\WToVal$)\,---\,which coincides with 
$\ToTot$ (\Resp $\WToTot$) in $\Cbn{\BangSet\!}$\,---\,is \emph{isomorphic} to \CbN\ (\Resp \ground\ \CbN) $\lambda$-calculus.
In particular, $\Cbn{(\cdot)}$ preserves normal forms (forth and back) and equates (via $\Tot$-equivalence) exactly \mbox{the same as $\beta$-equivalence.}

\newcounter{cor:cbn-preservation}
\addtocounter{cor:cbn-preservation}{\value{definition}}
\begin{corollary}[Preservations with respect to \CbN\ $\lambda$-calculus]
\label{cor:cbn-preservation}
  Let $\LTm, \LTmTwo \in \Lambda$.
  \begin{enumerate}
    \item\label{cor:cbn-preservation.equational-theory} \emph{\CbN\ equational theory:} $\LTm \simeq_\beta \LTmTwo$ iff $\Cbn{\LTm} \simeq_\ValScript \Cbn{\LTmTwo}$ iff $\Cbn{\LTm} \simeq_\Tot \Cbn{\LTmTwo}$.
    \item\label{cor:cbn-preservation.normal} \emph{\CbN\ normal forms:} $\LTm$ is (\ground) $\beta$-normal iff $\Cbn{\LTm}$ is (\ground) $\ValScript$-normal iff $\Cbn{\LTm}$ is (\ground) $\Tot$-normal.
  \end{enumerate}
\end{corollary}

\begin{proof}
	\begin{enumerate}
%
		\item 
		The equivalence ``$\Cbn{\LTm} \simeq_\ValScript \Cbn{\LTmTwo}$ iff $\Cbn{\LTm} \simeq_\Tot \Cbn{\LTmTwo}$'' holds because $\ToVal \ =\ \ToTot$ in $\Cbn{\BangSet\!}$, which is the image of $\Cbn{(\cdot)}$ (\Refrmk{cbn-image}).
		If $\LTm \simeq_\beta \LTmTwo$ then $\LTm \ToBeta^* \LTmThree \MRevTo{\beta} \LTmTwo$ for some $\LTmThree \in \Lambda$, as $\ToBeta$ is confluent;
		by \Refthmp{embedding}{cbn} (soundness), $\Cbn{\LTm} \ToVal^* \Cbn{\LTmThree} \MRevTo{\ValScript} \Cbn{\LTmTwo}$ and so $\Cbn{\LTm} \simeq_\ValScript \Cbn{\LTmTwo}$.
		Conversely, if $\Cbn{\LTm} \simeq_\Tot \Cbn{\LTmTwo}$ then $\Cbn{\LTm} \ToTot^* \TmThree \, \MRevTo{\Tot} \Cbn{\LTmTwo}$ for some $\TmThree \in \BangSet$, since $\ToTot$ is confluent (\Refpropp{confluence}{ToTot});
		by \Refthmp{embedding}{cbn} (completeness) and bijectivity of $\Cbn{(\cdot)}$ (\Refrmk{cbn-image}), $\LTm \ToBeta^* \LTmThree \MRevTo{\beta} \LTmTwo$ for some $\lambda$-term $\LTmThree$ such that $\Cbn{\LTmThree} = \TmThree$, and therefore $\LTm \simeq_\beta \LTmTwo$.

		\item Immediate consequence of \Refthmps{embedding}{cbn}{cbn-head}.
		\qedhere
	\end{enumerate}
\end{proof}

The correspondence between \CbV\ $\lambda$-calculus and bang calculus is slightly more delicate: \CbV\ translation 
$\Cbv{(\cdot)}$ gives a \emph{sound} and \emph{complete} embedding of $\ToBetav$ into $\ToBang\ToVal$ (and similarly for their \ground\ variants), but it is not complete with respect to generic $\ToTot$. 
Indeed, \Refex{delta} and \Refex{delta-translated} have shown that $\Cbv{(\La{\Var}{\Var\Var})} = \Bang{\Delta\!'} \ToBang \Bang{\Delta\!}$, where $\Bang{\Delta\!}$ is $\Tot$-normal and there is no $\lambda$-term $\LTm$ such that $\Cbv{\LTm} = \Bang{\Delta\!}$.
Note that $\La{\Var}{\Var\Var}$ is $\Betav$-normal but $\Cbv{(\La{\Var}{\Var\Var})} = \Bang{\Delta\!'}$ is not $\Tot$-normal: 
in \CbV\ the analogous of \Refcorp{cbn-preservation}{normal} does not hold for $\Cbv{(\cdot)}$.

Actually, an analogous of \Refcorp{cbn-preservation}{equational-theory} for \CbV\ holds: \CbV\ translation preserves $\Betav$-equivalence in a sound and complete way with respect to $\Tot$-equivalence (see \Refcor{cbv-preservation-equivalence} below).
The proof requires a fine analysis of \CbV\ translation $\Cbv{(\cdot)}$.
First, 
we define two subsets $\Cbv{\BangSet\!}$ and $\Cbv{\BangSet_v}$ of $\BangSet$, see \Reffig{targets}. 
%
%
\begin{remark}[Image of \CbV\ translation]
	\label{rmk:cbv-image}
	If $\LTm \in \Lambda$ then $\Cbv{\LTm} \in \Cbv{\BangSet\!}$; in particular, if $\LVal \in \LambdaVal$ then $\Cbv{\LVal} = \Bang{\ImVal}$ for some $\ImVal \in \Cbv{\BangSet_v}$.
	Note that $\Cbv{(\cdot)}$ is not surjective in $\Cbv{\BangSet\!}$: $ \Bang{\Delta\!} \in \Cbv{\BangSet\!}$ but there is no $\lambda$-term $\LTm$ such that $\Cbv{\LTm} = \Bang{\Delta\!}$.
\end{remark}
 
We then define a forgetful 
map $\InvV{(\cdot)} \colon \Cbv{\BangSet\!} \cup \Cbv{\BangSet_v} \to \Lambda$ 
from terms 
$\ImTmOne \in \Cbv{\BangSet\!}$ and $\ImVal \in \Cbv{\BangSet_v}$ 
into $\lambda$-terms:
\begin{align*}
  \InvV{(\Bang{\ImVal})} &\Defeq \InvV{\ImVal} & \InvV{(\App{\Der{\ImTmOne}}{\ImTmTwo})} &\Defeq \InvV{\ImTmOne}\InvV{\ImTmTwo} & \InvV{(\App{\ImVal}{\ImTmOne})} &\Defeq \InvV{\ImVal}\InvV{\ImTmOne}\,; && \quad&
  \InvV{\Var} &\Defeq \Var & \InvV{(\La{\Var}{\ImTmOne})} &\Defeq \La{\Var}{\InvV{\ImTmOne}}.
\end{align*}
%
\newcounter{lemma:forgetful}
\addtocounter{lemma:forgetful}{\value{definition}}
\begin{lemma}[Properties of the forgetful 
	map $\InvV{(\cdot)}$]
\label{lemma:forgetful}
\hfill
\NoteProof{lemmaAppendix:forgetful}
  \begin{enumerate}
    \item\label{lemma:forgetful.inverse} \emph{Forgetful map 
    	is a left-inverse of \CbV\ translation:} For every $\LTm \in \Lambda$, one has $\InvV{\Cbv{\LTm}} = \LTm$.
    \item\label{lemma:forgetful.substitution} 
    \emph{Substitution:} $\ImTmOne \Sub{\ImVal}{\Var} \in \Cbv{\BangSet\!}$ with $\InvV{(\ImTmOne \Sub{\ImVal}{\Var})} = \InvV{\ImTmOne} \Sub{\InvV{\ImVal}\!}{\Var}$, 
    for any $\ImTmOne \in \Cbv{\BangSet\!}$ and $\ImVal
    \in \Cbv{\BangSet_v}$.

    \item\label{lemma:forgetful.reduction} 
    \emph{$\Tot$-reduction vs.~$\Betav$-reduction:}
    For any $\ImTmOne \in \Cbv{\BangSet\!}$ and $\Tm \in \BangSet$, if $\ImTmOne \ToTot \Tm$ then $\Tm \in \Cbv{\BangSet\!}$ and $\InvV{\ImTmOne} \ToBetav^= \InvV{\Tm}$.
  \end{enumerate}
\end{lemma}

Despite the non-surjectivity of $\Cbv{(\cdot)}$ on $\Cbv{\BangSet}$, 
 \Reflemmap{forgetful}{reduction} and \Refrmk{cbv-image} mean that $\Cbv{\BangSet\!}$ is the set of terms in $\BangSet$ reachable by $\Tot$-reduction from \CbV\ translations of $\lambda$-terms (\Ie for any $\LTm \in \Lambda$, if $\Cbv{\LTm} \ToTot^* \TmTwo$ then 
$\TmTwo \in \Cbv{\BangSet\!}$);
moreover, $\Tot$-reduction on $\Cbv{\BangSet}$ is projected into $\Betav$-reduction on $\Lambda$ by the forgetful map $\InvV{(\cdot)}$.

%


\newcounter{cor:cbv-preservation-equivalence}
\addtocounter{cor:cbv-preservation-equivalence}{\value{definition}}
\begin{corollary}[Preservation of \CbV\ equational theory]
\label{cor:cbv-preservation-equivalence}
  Let $\LTm, \LTmTwo \in \Lambda$. 
  One has $\LTm \simeq_{\Betav} \LTmTwo$ iff $\Cbv{\LTm} \simeq_{\Tot} \Cbv{\LTmTwo}$.
\end{corollary}
\begin{proof}
	If $\LTm \simeq_{\Betav} \LTmTwo$ then $\LTm \ToBetav^* \LTmThree \, \MRevTo{\Betav\!} \LTmTwo$ for some $\LTmThree \in \Lambda$, as $\ToBetav$ is confluent;
	by \Refthmp{embedding}{cbv} (soundness), $\Cbv{\LTm} \ToTot^* \Cbv{\LTmThree} \MRevTo{\Tot} \Cbv{\LTmTwo}$ 
	and so $\Cbv{\LTm} \simeq_\Tot \Cbv{\LTmTwo}$.
	Conversely, if $\Cbv{\LTm} \simeq_{\Tot} \Cbv{\LTmTwo}$ then $\Cbv{\LTm} \ToTot^* \TmThree \,\MRevTo{\Tot} \Cbv{\LTmTwo}$ for some $\TmThree \in \BangSet$, since $\ToTot$ is confluent (\Refpropp{confluence}{ToTot});
	by \Refrmk{cbv-image}, $\Cbv{\LTm}\!, \Cbv{\LTmTwo} \in \Cbv{\BangSet\!}$;
	thus, $\TmThree \in \Cbv{\BangSet\!}$ and $\InvV{\Cbv{\LTm}} \ToBetav^* \InvV{\TmThree} \, \MRevTo{\Betav\!} \InvV{\Cbv{\LTmTwo}}$ by \Reflemmap{forgetful}{reduction}, hence $\LTm = \InvV{\Cbv{\LTm}} \simeq_{\Betav} \InvV{\Cbv{\LTmTwo}} = \LTmTwo$ by \Reflemmap{forgetful}{inverse}.
\end{proof}

So, \Refcor{cbv-preservation-equivalence} says that \CbV\ translation $\Cbv{(\cdot)}$\,---\,even if it is a sound but not complete embedding of $\Betav$-reduction into $\Tot$-reduction\,---\,is a sound and complete embedding of $\Betav$-equivalence into $\Tot$-equivalence.
Said differently, the non-completeness of the \CbV\ translation with respect to $\Tot$-reduction is just a syntactic detail, the \CbV\ translation (via $\Tot$-equivalence) 
equates exactly the same as $\Betav$-equivalence.

A final remark on the good rewriting properties of the bang calculus: the embeddings of \CbN\ and \CbV\ \lam-calculi into the bang calculus are \emph{finer} than the ones into the linear calculus \lincalc\ introduced in \cite{MaraistOderskyTurnerWadler99}. 
For instance, in \lincalc\ 
there is no fragment isomorphic to \CbN\ \lam-calculus;
also, the \CbV\ translation of the \lam-calculus into \lincalc\ is sound but not complete, and equates more than $\Betav$-equivalence (see \cite[Ex.17]{MaraistOderskyTurnerWadler99}).
Moreover, the bang calculus can be modularly extended with other reduction rules and/or syntactic constructs so that our \CbV\ translation embeds the extensions of \CbV\ $\lambda$-calculus studied in \cite{AccattoliGuerrieri16} into the corresponding extended version of the bang calculus, with results analogous to those presented here.

\section{\texorpdfstring{The bang calculus with respect to \CbN\ and \CbV\ $\lambda$-calculi, semantically}{The bang calculus with respect to \CbN\ and \CbV\ lambda-calculi, semantically}}
\label{sect:semantics}

The denotational models of the bang calculus we are interested in this paper are those induced by a denotational model of $\LL$. 
We recall the basic definitions and notations, see \cite{Mellies09,Ehrhard16,EhrhardG16} for more details.

\paragraph{Linear logic based denotational semantics of bang calculus.}
A denotational model of $\LL$ is given by:

\begin{itemize}
	\item A $*$-autonomous category $\cL$, namely a symmetric monoidal closed category $(\cL,\ITens,\One,\Leftu,\Rightu,\Assoc,\Sym)$ with a dualizing object $\Bot$. 
	We use $\Limpl XY$ for the linear exponential object, $\Evlin\in\cL(\Tens{(\Limpl XY)}{X},Y)$ for the evaluation morphism and $\Curlin$ for the linear currying map $\cL(\Tens ZX,Y)\to\cL(Z,\Limpl XY)$. 
	We use $\Orth X$ for the object $\Limpl X\Bot$ of $\cL$ (the \emph{linear  negation} of $X$). 
	This operation $\Orthp{\_}$ is a functor $\Op\cL\to\cL$.
	The category $\cL$ is cartesian with terminal object $\Top$, product $\IWith$, projections $\Proj i$ ($i \in \{1,2\}$). 
	By $*$-autonomy, $\cL$ is cocartesian with initial object $\Zero$, coproduct $\IPlus$ and injections $\Inj i$.
	
	\item 
	A functor $\Excl\_\colon \cL\to\cL$ which is:
	\begin{itemize}
		\item a comonad with counit $\Dercat X\in
		\cL(\Excl X,X)$ (\emph{dereliction}) and comultiplication $\Diggcat X\in 
		\cL(\Excl X,\Excl{\Excl X})$ (\emph{digging}), and 
		\item a strong symmetric monoidal functor---with Seely isos $\Seelyz \in \cL(1,\oc \Top)$ and $\Seelyt_{X,Y} \in \cL(\oc X \otimes \oc Y, \oc(X \with Y))$---from the symmetric monoidal category $(\cL,\IWith, \Top)$ to the symmetric
		monoidal category $(\cL,\ITens, 1)$, satisfying an additional coherence condition
		with respect to $\Diggcat{}$.
	\end{itemize} 
	\end{itemize}	

%
In order that $\cL$ is also a denotational model of the bang calculus we need a further assumption: 
\begin{equation}\label{eq:assumption}
\text{ 
	\emph{the unique morphism in $\cL(\Zero,\Top)$ 
	must be an iso} \ (to simplify, we assume just $\Zero = \Top$).}
\end{equation} 
From \eqref{eq:assumption} it follows that for any two objects $X$ and $Y$ there is a morphism $\Zerom XY \Defeq \mathsf{i}\Compl \mathsf{t}\in\cL(X,Y)$ where $\mathsf{t}$
is the unique morphism $X\to\Top$ and $\mathsf{i}$ is the unique morphism $\Zero\to
Y$. 
It turns out that this specific zero morphism satisfies the identities $f\, \Zerom XY =\Zerom XZ=\Zerom YZ g$ for all $f\in\cL(Y,Z)$ and $g\in\cL(X,Y)$ .
Assumption \eqref{eq:assumption} is satisfied by many models of $\LL$, like relational model \cite{BucciarelliEhrhard99}, finiteness spaces \cite{Ehrhard05}, Scott model \cite{Ehrhard12bis}, (hyper-)coherence \cite{Girard87,Ehrhard93} and probabilistic coherence spaces~\cite{DanosEhrhard08}, all models based on Indexed $\LL$~\cite{BucciarelliEhrhard99}.

A \emph{model} of the bang calculus is any object $\cU$ of $\cL$ satisfying the identity $\cU \cong \With{\Excl\cU}{(\Limpl{\Excl\cU}{\cU})}$ (we assume this iso to be an equality).
Note that this entails both $\oc\cU\retract \cU$ and $\Limpl{\oc\cU}{\cU}\retract \cU$.

Given a term $\Tm$ and a repetition-free list of variables $\Vect \Var=(\List \Var
1k)$ which contains all the free variables of $\Tm$, we can define a morphism
$\Psem{\Tm}_{\Vect \Var}\in\cL(\Tpower{(\Excl\cU)}k,\cU)$\,---\,the \emph{denotational 
semantics} (or \emph{interpretation}) of $\Tm$\,---\,where $(\oc \cU)^{\otimes k} \Defeq \bigotimes_{i=1}^k \oc\cU$. 
The definition is by induction on $\Tm \in \BangSet$:


\begin{itemize}
\item $\Psem{\Var_i}_{\seq \Var} \Defeq \Weakm{\cU}^{\otimes i-1}\otimes\ \Dercat{\cU}\otimes\Weakm{\cU}^{\otimes k-i}$ where $\Weakm{\cU}\in\cL(\Excl\cU,\One)$ is the weakening and we keep implicit the monoidality isos $\One\otimes \cU \simeq \cU$,

\item $\Psem{\La{\VarTwo}{\TmTwo}}_{\seq \Var} \Defeq \Pair{\Zerom{\Tpower{(\Excl\cU)}k}{\Excl\cU}}{\Curlin{(\Psem{\TmTwo}_{\Vect\Var,\VarTwo})}}$, where we assume without loss of generality 
$\VarTwo \notin \{\Var_1, \dots, \Var_k\}$,
    
\item $  \Psem{\App \TmTwo\TmThree}_{\Vect \Var} \Defeq \Evlin\Compl\Tensp{\Proj 2\Compl\Psem\TmTwo_{\Vect \Var}}{\Proj 1\Compl\Psem\TmThree_{\Vect \Var}}\Compl\Contr{},$ where $\Contr{} \in \cL(\Tpower{(\Excl\cU)}k, \Tens{\Tpower{(\Excl\cU)}k}{\Tpower{(\Excl\cU)}k})$ is the contraction,

\item $\Psem{\Bang \TmTwo}_{\Vect \Var} \Defeq \Pair{\Promp{\Psem \TmTwo_{\Vect \Var}}}{\Zerom{\Tpower{(\Excl\cU)}k}{\Limpl{\Excl\cU}{\cU}}}$, for $\Promp{\Psem \TmTwo_{\Vect \Var}}=\Excl{(\Psem S_{\Vect\Var})}\Compl\Coalgm{}$ where $\Coalgm{}\in\cL(\Tpower{(\Excl\cU)}k,\Excl{(\Tpower{(\Excl\cU)}k)})$ is the coalgebra structure map of $\Tpower{(\Excl\cU)}k$ (see \cite{Ehrhard16}),

\item $\Psem{\Der\TmTwo}_{\Vect\Var} \Defeq \Dercat{\cU}\Compl\Proj 1\Compl\Psem \TmTwo_{\Vect \Var}$\,.
\end{itemize}

\newcounter{thm:invariance}
\addtocounter{thm:invariance}{\value{definition}}
\begin{theorem}[Invariance, \cite{EhrhardG16}]
\label{thm:invariance}
  Let $\Tm, \TmTwo \in \BangSet$ and $\Vect{\Var}$
   be a repetition-free list of variables which contains all
  free variables of $\Tm$ and $\TmTwo$. If $\Tm 
  \simeq_\Tot \TmTwo$ then 
  $\Psem{\Tm}_{\Vect{\Var}} = \Psem{\TmTwo}_{\Vect{\Var}}$. 
\end{theorem}

The proof of \Refthm{invariance} uses crucially
the fact that $\Psem{\Bang{\TmThree}}_{\Vect{\Var}}$ is a coalgebra morphism, see \cite{Ehrhard16}.

The general notion of denotational model for the bang calculus presented here and obtained from \emph{any} denotational model $\cL$ of $\LL$ satisfying the assumption \eqref{eq:assumption} 
above is a particular case of Moggi's semantics of computations based on monads \cite{Moggi89,Moggi91}, 
if one keeps in mind that the functor ``$\oc$'' defines a strong monad on the Kleisli category $\cL_\oc$ of $\cL$.

\paragraph{Call-by-name.} 
A model of the \CbN\ \lam-calculus is a reflexive object in a cartesian closed category. The category $\cL$ being $*$-autonomous, its Kleisli $\Kl{\cL}$ over the comonad $(\oc,\Diggcat{},\Dercat{})$ is cartesian closed.
The category $\Kl{\cL}$ (whose objects are the same as $\cL$ and morphisms are given by $\Kl{\cL}(A,B) \Defeq \cL(\oc A,B)$) has composition 
$f \Klc g \Defeq f\,\oc g\,\Diggcat{}$ and identities 
$A \Defeq \Dercat{A}$.
In $\Kl{\cL}$, products $A \with B$ are preserved, 
with projections $\pi_i \Compl \Defeq \Proj{i}\Dercat{\oc(A \with B)}$ ($i \in \{1,2\}$); 
the exponential object $\Impl{A}{B}$ is $\Limpl{\oc A}{B}$ (this is the semantic counterpart of Girard's \CbN\ translation) and has an evaluation morphism $\Ev \Defeq \Evlin (\Dercat{\Limpl{\oc A}{B}}\otimes\Id_{\oc A}) (\Seelyt)^{-1} \in \cL(\oc((\Limpl{\oc A}{B}) \with \oc A),B)$.
This defines an exponentiation since for all $f \in\cL(\oc(C \with A), B) 
$ there is a unique morphism \mbox{$\Cur f \Defeq \Curlin(f\, \Seelyt) \in \cL(\oc C, \Limpl{\oc A}{B})$ satisfying $\Ev\Klc \Pair{\Cur{f}}{A} = f$.}

The identity $\cU = \With{\Excl\cU}{(\Limpl{\Excl\cU}{\cU})}$ satisfied by our object $\cU$ (the model of the bang calculus) entails $\Limpl{\Excl\cU}{\cU}\retract \cU$ in $\cL$ via $\Lam{} \Defeq \Pair{\Zerom{\Limpl{\Excl\cU}{\cU}}{\oc\cU}}{\Id_{\Limpl{\Excl\cU}{\cU}}} \in \cL(\Limpl{\oc\cU}{\cU},\cU)$ and $\app{} \Defeq \Proj 2 \in \cL(\cU, \Limpl{\oc\cU}{\cU})$, 
since $\app{} \Compl\, \Lam{} = \Id_{\Limpl{\oc\cU}{\cU}}$.
So, $\cU$ is a reflexive object (\Ie~$\Limpl{\Excl\cU}{\cU}\retract \cU$) in $\Kl\cL$ via $\app{n} \Defeq \Dercat{\Limpl{\oc\cU}{\cU}}\oc\app{} \in \Kl{\cL}(\cU, \Limpl{\oc\cU}{\cU})
$ and $\Lam{n} \Defeq \Dercat{\cU}\oc\Lam{} \in \Kl{\cL}(\Limpl{\oc\cU}{\cU}, \cU)
$, since $\app{n} \Klc \Lam n = \Limpl{\oc\cU}{\cU}$.
Then, the interpretation of a \lam-term $\LTm$ can be defined, as usual, as a morphism $\Intn{\LTm}_{\seq{\Var}} \in\Kl\cL(\cU^{k},\cU)$, with $\seq{\Var} = (\Var_1, \dots, \Var_k)$ such that $\Fv{\LTm} \subseteq \{\Var_1, \dots, \Var_k\}$ and $\Var_i \neq \Var_j$:
\begin{align*}
  \Intn{\Var_i}_{\seq \Var} &\Defeq 
  \pi_i^k, &
  \Intn{\La{\!\VarTwo\,}{\LTm}}_{\seq \Var} &\Defeq \Lam{n} \Klc \Cur{\Intn{\LTm}_{\seq{\Var},\VarTwo}}, &
  \Intn{\LTm\LTmTwo}_{\seq \Var} &\Defeq \Ev \Klc \Pair{\app{n} \Klc \Intn{\LTm}_{\seq \Var}}{\Intn{\LTmTwo}_{\seq \Var}}.
\end{align*}

Summing up, the object $\cU$ provides both a model of the bang calculus and a model of the \CbN\ \lam-calculus. 
The relation between the two is elegant:
the semantics $\Intn{\LTm}$ in the \CbN\ model of the \lam-calculus of a $\lambda$-term $\LTm$ decomposes into the semantics $\Psem{\cdot}$ in the model of the bang calculus of \mbox{the \CbN\ translation $\Cbn{\LTm}$ of $\LTm$.}

\newcounter{thm:semantics-cbn}
\addtocounter{thm:semantics-cbn}{\value{definition}}
\begin{theorem}[Factorization of any \CbN\ semantics]
\label{thm:semantics-cbn}
  For every \lam-term $\LTm$ and every repetition-free list $\seq{\Var} = (\Var_1, \dots, \Var_k)$ of variables such that $\Fv{\LTm} \subseteq \{\Var_1, \dots, \Var_k\}$, one has $\Psem{\Cbn \LTm}_{\vec \Var} 
  = \Intn{\LTm}_{\Vect \Var}$ (up to Seely's isos
  ).
\end{theorem}

\begin{proof} 
	Below we use $\cong$ to transform a morphism $f\in\cL((\oc \cU)^{\otimes k},\cU)$ into a morphism $g\in\Kl\cL(\cU^{k},\cU)$ using Seely's isos (where $\cU^{k} \Defeq \bigwith_{i=1}^k \cU$).
	We proceed by 
	induction on $\LTm \in \Lambda$.
	
		If $\LTm$ is a variable,  then $\LTm = \Var_i = \Cbn{\LTm}$, so $\Psem{\Cbn{\LTm}}_{\seq \Var} = \Psem{\Var_i}_{\seq \Var} = \Weakm{\cU}^{\otimes i-1}\otimes\ \Dercat{\cU}\otimes\Weakm{\cU}^{\otimes k-i} \cong \Proj{i}^k \Dercat{\cU^k} = \Intn{\Var_i}_{\seq \Var}$. 

		If $\LTm \Defeq \La{\VarTwo}{\LTmTwo}$ then $\Cbn{\LTm} = \La{\VarTwo}{\Cbn{\LTmTwo}}$.
		Hence,
		\begin{align*}
		\Psem{\La{\VarTwo}{\Cbn{\LTmTwo}}}_{\seq x} &= \Pair{\Zerom{\Tpower{(\Excl\cU)}k}{\Excl\cU}}{\Curlin{(\Psem{\Cbn{\LTmTwo}}_{\Vect\Var,\VarTwo})}} \ 
		\cong \ \Pair{\Zerom{(\Excl\cU)^{\otimes k}}{\oc\cU}}{\Curlin(\Intn{\LTmTwo}_{\Vect \Var,\VarTwo}\Seelyt)}\\
		&= \Dercat{\cU} \oc (\Pair{\Zerom{(\Excl\cU)^{\otimes k}}{\oc\cU}}{\Curlin(\Intn{\LTmTwo}_{\Vect \Var,\VarTwo}\Seelyt)})\Diggcat{\cU^k}\\
		&= \Dercat{\cU} \oc \Pair{\Zerom{\Limpl{\Excl\cU}{\cU}}{\oc\cU}}{\Id_{\Limpl{\Excl\cU}{\cU}}}\,\oc{\Curlin(\Intn{\LTmTwo}_{\Vect \Var,\VarTwo}\Seelyt)}\Diggcat{\cU^k}\\
		&= (\Dercat{\cU} \oc{\Lam{}})\Klc\Cur{\Intn{\LTmTwo}_{\Vect \Var,\VarTwo}}= \Intn{\La{\VarTwo}{\LTmTwo}}_{\Vect \Var} \ .
		\end{align*}
		
		If $\LTm \Defeq \LTmTwo\LTmThree$ then $\Cbn{\LTm} = \App{\Cbn{\LTmTwo}}{\Bang{(\Cbn \LTmThree)}}$.
		Therefore,
		\begin{align*}
		\Psem{\App{\Cbn{\LTmTwo}}{\Bang{(\Cbn \LTmThree)}}}_{\seq x} 
		&= \Evlin ((\Proj 2\, \Psem{\Cbn \LTmTwo}_{\Vect \Var})\otimes(\Proj 1\, \Psem{\Bang{(\Cbn \LTmThree)}}_{\Vect \Var}))\Contr{} \ 
		= \Evlin ((\Proj 2\, \Psem{\Cbn \LTmTwo}_{\Vect \Var})\otimes(\Proj 1\, \Pair{\Promp{\Psem {\Cbn \LTmThree}_{\Vect \Var}}}{\Zerom{\Tpower{(\Excl\cU)}k}{\Limpl{\Excl\cU}{\cU}}}))\Contr{}\\
		&= \Evlin ((\Proj 2\, \Psem{\Cbn \LTmTwo}_{\Vect \Var})\otimes\Promp{\Psem {\Cbn \LTmThree}_{\Vect \Var}})\Contr{} \ 
		\cong \ \Evlin (
		(\app{} \Compl \Intn{\LTmTwo}_{\Vect \Var})
		\otimes
		(\oc\Intn{\LTmThree}_{\Vect \Var}\,\Diggcat{\cU^k})
		)\Contr{}\\
		&= \Evlin (\Dercat{\Limpl{\oc \cU}{\cU}}\oc(\app{} \Compl \Intn{\LTmTwo}_{\Vect \Var})\Diggcat{\cU^k})
		\otimes(\oc\Intn{\LTmThree}_{\Vect \Var}\,\Diggcat{\cU^k}))\Contr{}\\
		&= \Evlin (\Dercat{\Limpl{\oc \cU}{\cU}}\otimes\Id_{\oc \cU})\,(\oc(\app{} \Compl \Intn{\LTmTwo}_{\Vect \Var})\otimes\oc\Intn{\LTmThree}_{\Vect \Var})(\Diggcat{\cU^k}\otimes\Diggcat{\cU^k})\Contr{}\\
		&= \Evlin (\Dercat{\Limpl{\oc \cU}{\cU}}\otimes\Id_{\oc \cU})\Seelyt\oc\Pair{\app{} \Compl \Intn{\LTmTwo}_{\Vect \Var}}{\Intn{\LTmThree}_{\Vect \Var} }\Diggcat{\cU^k}\\
		&= \Ev\Klc\Pair{\Dercat{\Limpl{\oc\cU}{\cU}}\,\oc(\app{} \Compl\,\Intn{\LTmTwo}_{\Vect \Var})\Diggcat{\cU^k} }{\Intn{\LTmThree}_{\Vect \Var} } \ 
		= \ \Ev\Klc\Pair{\Dercat{\Limpl{\oc\cU}{\cU}}\,\oc{\app{}} \Compl \oc\Intn{\LTmTwo}_{\Vect \Var}\Diggcat{\cU^k} }{\Intn{\LTmThree}_{\Vect \Var} } \\
		&= \Ev\Klc\Pair{(\Dercat{\Limpl{\oc\cU}{\cU}}\,\oc{\app{}})\Klc\Intn{\LTmTwo}_{\Vect \Var} }{\Intn{\LTmThree}_{\Vect \Var} } = \Intn{\LTmTwo\LTmThree}_{\Vect \Var} \ .
		\hfill\qedhere
		\end{align*}		
\end{proof}

\Refthm{semantics-cbn} is a powerful result: it says not only that \emph{every} $\LL$ based model of the bang calculus is also a model of the \CbN\ $\lambda$-calculus, but also that the \CbN\ semantics of any $\lambda$-term in such a model \emph{always} naturally factors into the \CbN\ translation of the $\lambda$-term and its semantics in the bang calculus.

\paragraph{Call-by-value.} 
Following~\cite{PravatoRonchiRoversi99,Ehrhard12}, models of the \CbV\ \lam-calculus can be defined using 
Girard's ``boring'' \CbV\ translation of the intuitionistic 
implication into $\LL$.
It is enough to find an object $X$ in $\cL$ satisfying $\Limpl{\oc X}{\oc X}\retract X$ (or equivalently, $\oc(\Limpl{X}{X})\retract X$).\footnotemark
\footnotetext{This approach is compatible with other notions of model such as Moggi's one \cite{Moggi89,Moggi91}, 
since the functor ``$\oc$'' defines a strong monad on the Kleisli category $\Kl\cL$.
The reflexive object $\Limpl{\oc X}{\oc X}\retract X$ (or equivalently, $\oc(\Limpl{X}{X})\retract X$) is the \CbV\ version in $\cL$ of the reflexive object $X \Rightarrow X \retract X$ in a cartesian closed category, in accordance with Girard's \CbV\ decomposition of the arrow.}	
This is the case for our object $\cU$ (the model of the bang calculus) since $\Limpl{\oc\cU}{\cU}\retract \cU$ and $\oc\cU \retract \cU$ entail $\Limpl{\oc\cU}{\oc\cU}\retract\cU$ (by the variance of $\oc\cU\multimap\_$) via the morphisms $\Lam v = \Pair{\Zerom{\Limpl{\oc\cU}{\oc\cU}}{\oc\cU}}{\Curlin(\Pair{\Evlin}{\Zerom{(\Limpl{\oc\cU}{\oc\cU})\otimes\oc\cU}{\cU}})} \in \cL(\Limpl{\oc\cU}{\oc\cU}, \cU)$ and $\app v =\Curlin(\Proj 1\,\Evlin)\Proj 2 \in \cL(\cU, \Limpl{\oc\cU}{\oc\cU})$.
As in \cite{Ehrhard12}, we can then define the interpretation of a \lam-term $\LTm$ as a morphism $\Intv{\LTm}_{\seq \Var} \in\cL((\oc\cU)^{\otimes k},\oc\cU)$, with $\seq{\Var} = (\Var_1, \dots, \Var_k)$ such that $\Fv{\LTm} \subseteq \{\Var_1, \dots, \Var_k\}$ and $\Var_i \neq \Var_j$:
\begin{align*}
  \Intv{\Var_i}_{\seq \Var} &= \Weakm{\cU}^{\otimes i-1}\otimes\ \Id_{\oc\cU}\otimes\Weakm{\cU}^{\otimes k-i}, &
  \Intv{\La{\!\VarTwo\,}{\LTm}}_{\seq \Var} &= \Bang{(\Lam v\Compl\Curlin(\Intv{\LTm}_{\seq{\Var},\VarTwo}))}, &
  \Intv{\LTm\LTmTwo}_{\seq \Var} &= \Evlin ((\app v\Compl \Intv{\LTm}_{\seq \Var})\otimes(\Intv{\LTmTwo}_{\seq \Var}))\Contr{}.
\end{align*}
We now have two possible ways of interpreting the \CbV\ \lam-calculus in our model $\cU$: either by translating a \lam-term $\LTm$ into 
$\Cbv{\LTm} \in \BangSet$ and then compute $\Psem{\Cbv \LTm}$, or by computing directly~$\Intv{\LTm}$.
It is natural to wonder whether the two interpretations $\Psem{\Cbv \LTm}$ and $\Intv{\LTm}$ are related, and in what way. In~\cite{EhrhardG16} the authors conjectured that, at least in the case of a particular relational model $\cU$ satisfying $\cU = \oc\cU\cup (\oc{\cU}\times \cU) = \oc \cU\with (\oc\cU\multimap \cU)$, the two interpretations coincide.
We show that the situation is actually more complicated than expected.

The \emph{relational model} $\cU$ introduced in~\cite{EhrhardG16} admits the following concrete description as a type system.
The set $\cU$ of \emph{types} and the set $\oc \cU$ of \emph{finite multisets} over $\cU$ are 
defined by mutual induction as follows:
\begin{align}
\label{eq:types}
	\textup{(set: }\cU\textup{)}\quad \alpha,\beta,\gamma &\Coloneqq a\mid a\lto\alpha\qquad \qquad& \textup{(set: }\oc \cU\textup{)}\quad
	a,b,c &\Coloneqq [\alpha_1,\dots,\alpha_k]\quad\textup{ for any } k \ge 0.
\end{align}
\emph{Environments} $\Gamma$ are functions from variables to $\oc \cU$ 
such that $\Supp{\Gamma} \Defeq \{ \Var \in \VarSet \mid \Gamma(x)\neq \Emptymset\}$ is finite.
We write $\Var_1 \!:a_1,\dots, \Var_k \!:a_k$ for the environment $\Gamma$ satisfying $\Gamma(\Var_i) = a_i$ and $\Gamma(\VarTwo) = \Emptymset$ for $y\notin\{\Var_1, \dots, \Var_k\}$.
The multiset union $a + b$ is extended to environments pointwise, namely $(\Gamma+\Delta)(\Var) \Defeq \Gamma(\Var)+\Delta(\Var)$.

On the one hand (see \cite{Ehrhard12,AccattoliGuerrieri18}), 
the relational model $\cU$ for the \CbV\ \lam-calculus interprets a \lam-term $\LTm$ using $\Intv{\cdot}$, which gives 
$\Intv{\LTm}_{\seq \Var} = \{ (a_1,\dots, a_k,\beta) \mid \Var_1 \!:\! a_1,\dots, \Var_k \!:\! a_k \Vdashv \LTm \!:\! \beta \text{ is derivable} \}$ where $\seq{\var} = (\Var_1,\dots, \Var_k)$ with  $\Fv{\LTm} \subseteq \{\Var_1, \dots, \Var_k\}$, and $\Vdashv$ is the type system below 
(note that if $\Gamma \Vdashv \LTm : \beta$ is derivable then $\beta \in \oc \cU$):
\begin{equation*}
  \begin{prooftree}[rule margin = 0.4ex, center=false, label separation = 0.4em]
	\infer0[\footnotesize{ax}]{\Var : a \Vdashv \Var : a}
	\end{prooftree}
	\qquad
	\begin{prooftree}[rule margin = 0.4ex, center=false, label separation = 0.4em]
	\hypo{\Gamma\Vdashv \LTm : [a\!\lto\! b] }
	\hypo{\Delta\Vdashv \LTmTwo : a }    
	\infer2[\footnotesize{$\mathrm{app}$}]{\Gamma + \Delta\Vdashv \LTm\LTmTwo : b}
	\end{prooftree}
	\qquad  
	\begin{prooftree}[rule margin = 0.4ex, center=false, label separation = 0.4em]
	\hypo{(\Gamma_i, \VarTwo \!:a_i \Vdashv \LTm : b_i)_{1 \leq i \leq k} \quad k \ge 0}
	\infer1[\footnotesize{$\mathrm{lam}$}]{\sum_{i=1}^k \Gamma_i \Vdashv \La{\!\VarTwo\,}{\LTm} \!: [a_1\!\lto\! b_1,\dots,a_k\!\lto\! b_k]}
	\end{prooftree}.
\end{equation*}

On the other hand (see \cite{EhrhardG16}), 
the relational model $\cU$ for the bang calculus interprets a term $\Tm \in \BangSet$ using $\Psem{\cdot}$, which gives 
$\Psem{\Tm}_{\seq \Var} = \{ (a_1,\dots, a_k, \beta) \mid \Var_1 \!:\! a_1,\dots, \Var_k \!:\! a_k \vdash_\oc \Tm \!:\! \beta \text{ is derivable} \}$ where  $\seq{\var} = (\Var_1,\dots, \Var_k)$ with  $\Fv{\LTm} \subseteq \{\Var_1, \dots, \Var_k\}$, and $\vdash_\oc$ is the following type system:
$$
\begin{prooftree}[label separation = 0.2em, center=false]
\infer0[\footnotesize{$ax_\oc$}]{\Var \!: [\alpha] \vdash_\oc \Var \!: \alpha}
\end{prooftree}
\qquad
\begin{prooftree}[label separation = 0.2em, rule margin = 0.4ex, center=false]
\hypo{\Gamma\vdash_\oc \Tm \!: a\!\lto\! \beta }
\hypo{\Delta\vdash_\oc \TmTwo \!: a }    
\infer2[\footnotesize{$\mathrm{@}$}]{\Gamma + \Delta\vdash_\oc \App{\Tm}{\TmTwo} \!: \beta}
\end{prooftree}
\qquad  
\begin{prooftree}[label separation = 0.2em, rule margin = 0.4ex, center=false]
\hypo{(\Gamma_i \vdash_\oc \Tm \!: \beta_i)_{1 \leq i \leq k} \quad k \ge 0}
\infer1[\footnotesize{$\oc$}]{\sum_{i=1}^k \Gamma_i \vdash_\oc \Bang{\Tm} \!: [\beta_1,\dots, \beta_k]}
\end{prooftree}
$$

\vspace{-.8\baselineskip}
$$
\begin{prooftree}[label separation = 0.2em, rule margin = 0.4ex, center=false]
\hypo{\Gamma \vdash_\oc \Tm \!: [\alpha]}
\infer1[\footnotesize{$\Der$}]{\Gamma  \vdash_\oc \Der{\Tm} \!: \alpha}
\end{prooftree}
\qquad
\begin{prooftree}[label separation = 0.2em, rule margin = 0.4ex, center=false]
\hypo{\Gamma, \Var \colon\! a \vdash_\oc \Tm \!: \beta}
\infer1[\footnotesize{$\lambda$}]{\Gamma \vdash_\oc \La{\Var}{\Tm} \!: a \lto \beta}
\end{prooftree}.
$$

In $\cU$ (seen as the relational model for the bang calculus) 
what is the interpretation 
$\Psem{\Cbv{\LTm}}_{\seq \Var}$ of the \CbV\ translation $\Cbv{\LTm}$  
of a $\lambda$-term $\LTm$?
Easy calculations show that in the type system $\vdash_\oc$ the rules below\,---\,
the ones needed to interpret terms of the form $\Cbv{\LTm}$ for some $\lambda$-term $\LTm$\,---\,can be derived: 
\begin{equation}
\label{eq:inference-cbv}
  \begin{prooftree}[label separation = 0.2em]
    \infer0[\footnotesize{ax}]{\Var \!: a \vdash_\oc \Bang{\Var} \!: a}
  \end{prooftree}
  \quad \
  \begin{prooftree}[label separation = 0.2em, rule margin = 0.4ex, separation = 1.45em]
    \hypo{\Gamma\vdash_\oc \Cbv{\LTm} \!: [a\!\lto\! \beta] }
    \hypo{\Delta\vdash_\oc \Cbv{\LTmTwo} \!: a }    
    \infer2[\footnotesize{$\mathrm{app}$}]{\Gamma + \Delta\vdash_\oc \App{\Der{\,\Cbv{\LTm}}}{\Cbv{\LTmTwo}} \!: \beta}
  \end{prooftree}
  \quad \
  \begin{prooftree}[label separation = 0.2em, rule margin = 0.4ex	]
    \hypo{(\Gamma_i, \VarTwo \!:a_i \vdash_\oc \Cbv{\LTm} \!: \beta_i)_{1 \leq i \leq k} \quad k \ge 0}
    \infer1[\footnotesize{$\mathrm{lam}$}]{\sum_{i=1}^k \Gamma_i \vdash_\oc \Bang{(\La{\!\VarTwo\,}{\Cbv{\LTm}})} \!: [a_1\!\lto\!\beta_1,\dots,a_k\!\lto\!\beta_k]}
  \end{prooftree}
\end{equation}

\noindent Intuitively, the type system $\Vdashv$ is obtained from the restriction of $\vdash_\oc$ 
to the image of $\Cbv{(\cdot)}$ by substituting arbitrary types $\beta$ with multisets $b$ of types. 
So, given a $\lam$-term $\LTm$, the two interpretations $\Intv{\LTm}_{\seq \Var}$ and $\Psem{\Cbv{\LTm}}_{\seq \Var}$ can be different: 
for $\alpha \in \cU \smallsetminus \oc \cU$ (\Eg take $\alpha = \emptymset \lto \emptymset$), one has $[([a \!\lto\! \alpha] + a)\lto \alpha]\in\Psem{\Cbv{(\La{\Var}{\Var\Var})}} \smallsetminus \Intv{\La{\Var}{\Var\Var}}$.

\begin{proposition}[Relational semantics for \CbV]
In the relational model $\cU$, $\Intv{\LTm}_{\seq \Var} \subseteq \Psem{\Cbv{\LTm}}_{\seq \Var}$ for any $\lambda$-term $\LTm$, with $\seq{\Var} = (\Var_1, \dots, \Var_k)$ such that $\Fv{\LTm} \subseteq \{\Var_1, \dots, \Var_k\}$.
There exists a closed $\lambda$-term $\LTmTwo$ such that $\Intv{\LTmTwo} \neq \Psem{\Cbv {\LTmTwo}}$.
\end{proposition}

\begin{proof}
	We have just shown that $\Intv{\LTmTwo} \neq \Psem{\Cbv {\LTmTwo}}$ for $\LTmTwo = \La{\Var}{\Var\Var}$.
	To prove that $\Intv{\LTm}_{\seq \Var} \subseteq \Psem{\Cbv{\LTm}}_{\seq \Var}$, it is enough to show, by induction on $\LTm \in \Lambda$, that $\Var_1 \!:\! a_1,\dots, \Var_n \!:\! a_k \vdash_\oc \Cbv{\LTm} \!:\! \beta $ is derivable whenever $\Var_1 \!:\! a_1,\dots, \Var_k \!:\! a_k \Vdashv \LTm \!:\! \beta $ is.
	
	If $\LTm$ is a variable, then $\LTm = \Var_i$ for some $1 \leq i \leq k$, and $\Cbv{\LTm} = \Bang{\Var_i}$. 
	All derivations for $\LTm$ in the type system $\Vdashv$ are of the form    
	\begin{prooftree}[rule margin = 0.4ex, center=false, label separation= 0.2em]
		\infer0[\footnotesize{ax}]{\Var_i : a \Vdashv \Var_i : a}
	\end{prooftree}
	for any $a \in \oc\cU$, and in the type system $\vdash_\oc$, according to \eqref{eq:inference-cbv},
	\begin{prooftree}[label separation = 0.2em, rule margin = 0.4ex, center=false]
		\infer0[\footnotesize{ax}]{\Var_i \!: a \vdash_\oc \Bang{\Var_i} \!: a}
	\end{prooftree}
	is derivable.
	
	If $\LTm = \LTmTwo\LTmThree$, then $\Cbv{\LTm} = \App{\Der{\,\Cbv{\LTmTwo}}}{\Cbv{\LTmThree}}$ and all derivations for $\LTm$ in the type system $\Vdashv$ are of the form
	\begin{equation*}
		\begin{prooftree}[rule margin = 0.4ex]
		\hypo{\Gamma\Vdashv \LTmTwo \!: [a\!\lto\! b] }
		\hypo{\Delta\Vdashv \LTmThree \!: a }    
		\infer2[\footnotesize{$\mathrm{app}$}]{\Gamma + \Delta\Vdashv \LTmTwo\LTmThree \!: b}
		\end{prooftree}
		\qquad \text{for any } a,b \in \oc\cU.
	\end{equation*}
	By \Ih, $\Gamma\vdash_\oc \Cbv{\LTmTwo} \!: [a\!\lto\! \beta]$ and $\Delta\vdash_\oc \Cbv{\LTmTwo} \!: a$ are derivable in the type system $\vdash_\oc$, hence the following derivation is derivable in the type system $\vdash_\oc$, according to \eqref{eq:inference-cbv} since $\oc\cU \subseteq \cU$
	\begin{equation*} 
	\begin{prooftree}[rule margin = 0.4ex, label separation = 0.2em, center=false]
		\hypo{\Gamma\vdash_\oc \Cbv{\LTmTwo} \!: [a\!\lto\! b] }
		\hypo{\Delta\vdash_\oc \Cbv{\LTmThree} \!: a }    
		\infer2[\footnotesize{$\mathrm{app}$}]{\Gamma + \Delta\vdash_\oc \App{\Der{\,\Cbv{\LTmTwo}}}{\Cbv{\LTmThree}} \!: b}
	\end{prooftree}.
	\end{equation*}
	
	If $\LTm = \La{\VarTwo}{\LTmTwo}$, then $\Cbv{\LTm} = \Bangp{\La{\VarTwo}{\Cbv{\LTmTwo}}}$ and all derivations for $\LTm$ in the type system $\Vdashv$ are of the form
	\begin{equation*}
	\begin{prooftree}[rule margin = 0.4ex]
	\hypo{(\Gamma_i, \VarTwo \!:a_i \Vdashv \LTmTwo \!: b_i)_{1 \leq i \leq k} \quad k \ge 0}
	\infer1[\footnotesize{$\mathrm{lam}$}]{\sum_{i=1}^k \Gamma_i \Vdashv \La{\!\VarTwo\,}{\LTmTwo} \!: [a_1\!\lto\! b_1,\dots,a_k\!\lto\! b_k]}
	\end{prooftree}
	\qquad \text{for any } a_1, b_1, \dots, a_k, b_k \in \oc\cU.
	\end{equation*}
	By \Ih, $\Gamma_i, \VarTwo \!: a_i \vdash_\oc \Cbv{\LTmTwo} \!: b_i$ is derivable in the type system $\vdash_\oc$ for all $1 \leq i \leq k$, hence the following derivation is derivable in the type system $\vdash_\oc$, according to \eqref{eq:inference-cbv} since $\oc\cU \subseteq \cU$
	\begin{equation*}
	\begin{prooftree}[rule margin = 0.4ex, label separation = 0.2em, center=false]
	\hypo{(\Gamma_i, \VarTwo \!:a_i \vdash_\oc \Cbv{\LTmTwo} \!: b_i)_{1 \leq i \leq k}\quad k \ge 0}
	\infer1[\footnotesize{$\mathrm{lam}$}]{\sum_{i=1}^k \Gamma_i \vdash_\oc \Bang{(\La{\!\VarTwo\,}{\Cbv{\LTmTwo}})} \!: [a_1\!\lto\!b_1,\dots,a_k\!\lto\!b_k]}
	\end{prooftree}.
	\qedhere
	\end{equation*}
\end{proof}

The example above of $\Intv{\LTmTwo} \neq \Psem{\Cbv {\LTmTwo}}$ shows also that in general neither $\Psem{\Cbv \LTm}_{\seq \Var} = \Pair{\Intv{\LTm}}{0}_{\seq \Var}$ nor $\Proj 1\Psem{\Cbv \LTm}_{\seq \Var} = \Intv{\LTm}_{\seq \Var}$ hold in relational semantics.\footnotemark
\footnotetext{Relational semantics interprets terms in the object $\cU$\,---\,defined in \eqref{eq:types}, where $a\lto\alpha$ denotes the ordered pair $(a,\alpha)$\,---\,of the category $\mathbf{Rel}$ of sets and relations.
The cartesian product $\with$ is the disjoint union, with the empty set as terminal and initial object $\top = 0$, so that the zero morphism $0_{X,Y}$ for any objects $X$ and $Y$ is the empty relation and the projection $\Proj{i}$ is the obvious selection.
Therefore, in relational semantics, $\Pair{\Intv{\LTm}_{\seq \Var}}{0} = \Intv{\LTm}_{\seq \Var}$ and $\Proj 1\Psem{\Cbv \LTm}_{\seq \Var} = \Psem{\Cbv{\LTm}}_{\seq \Var}$ for any $\lambda$-term $\LTm$.
}
We 
conjecture that, for any $\lambda$-term $\LTm$, $\Intv{\LTm}_{\seq \Var}$ can be obtained from $\Psem{\Cbv \LTm}_{\seq \Var}$ by iterating the application of $\Proj 1$ to $\Psem{\cdot}$ along the structure of $\LTm$, but how to express this formally and categorically for a generic model $\cU$ of the bang calculus?
Usually in these situations one defines a logical relation between the two interpretations, but this is complicated by the fact that we are in the untyped setting so there is no type hierarchy to base our induction.
We plan to investigate whether the (syntactic) logical relations introduced by Pitts in~\cite{Pitts93} can give an inspiration to define semantic logical relations in the untyped setting. 
Another source of inspiration might be the study of other concrete $\LL$ based models of the \CbV\ $\lambda$-calculus, such as Scott domains and coherent semantics \cite{PravatoRonchiRoversi99, Ehrhard12}.

 %

\section{Conclusions}
\label{sect:conclusions}

The bang calculus is a general setting to study and compare \CbN\ and \CbV\ $\lambda$-calculi in the \emph{same} rewriting system and with the \emph{same} denotational semantics, as we have shown.
Since \CbN\ and \CbV\ \lam-calculi are usually investigated as \emph{two} different rewriting systems with \emph{two} distinct semantics, the study of the bang calculus can be fruitful 
because it provides a more general, canonical and unifying setting where:
\begin{itemize}
	\item  operational and denotational notions and properties (such as models, continuations, standardization, normalization strategies, equational theories induced by denotational models, \Etc) can be introduced and investigated, 
	so that one can obtain their \CbN\ and \CbV\ counterparts by just restricting the general notion or result for the bang calculus to the \CbN\ and \CbV\ fragments of the bang calculus;
	\item in particular, many well studied theoretical notions of the \CbN\  $\lambda$-calculus that do not have satisfactory \CbV\ counterparts yet (such as	separability, solvability, B\"ohm trees, classification of $\lambda$-theories, full-abstraction, \Etc) might be generalized in the bang calculus, so as to obtain their \CbV\ counterparts when restricted to the \CbV\ fragment of the bang calculus.
\end{itemize}

\paragraph{Acknowledgment.} The authors are grateful to Thomas Ehrhard and Claudia Faggian for helpful comments and stimulating discussions.

\bibliographystyle{eptcs}
\bibliography{biblio}
\addcontentsline{toc}{section}{References}

%
\end{document}